\documentclass[10pt]{article}

\usepackage{amsfonts}
\usepackage{amsmath}
\usepackage{amssymb}
\usepackage{amsthm}
\newtheorem{theorem}{Theorem}
\newtheorem{lemma}[theorem]{Lemma}

\newtheorem{definition}{Definition}

\usepackage{amssymb}
\usepackage{bm}
\newcommand{\bu}{\bm{u}}
\newcommand{\bv}{\bm{v}}
\newcommand{\bw}{\bm{w}}

\title{Deriving bases for Abelian functions}
\author{Matthew England}
\date{May 2011}

\begin{document}

\maketitle

\begin{abstract}
We present a new method to explicitly define Abelian functions associated with algebraic curves, for the  purpose of finding bases for the relevant vector spaces of such functions.  We demonstrate the procedure with the functions associated with a trigonal curve of genus four.  The main motivation for the construction of such bases is that it allows systematic methods for the derivation of the addition formulae and differential equations satisfied by the functions.  We present a new 3-term 2-variable addition formulae and a complete set of differential equations to generalise the classic Weierstrass identities for the case of the trigonal curve of genus four.
\end{abstract}

\bigskip

\subsubsection*{Contact Details}
School of Mathematics and Statistics, University of Glasgow, G12 8QQ, UK \\ 
Matthew.England@gla.ac.uk \\ 
\texttt{http://www.maths.glasgow.ac.uk/\~{}mengland/}

\subsubsection*{Acknowledgements}
Thanks to Chris Eilbeck, Yoshihiro Onishi and Chris Athorne for useful conversations on this topic.  Particular thanks goes to one of the anonymous referees for useful suggestions throughout the paper and especially those comments that led to the discussion of Abelian function as operators on sigma functions at the end of Section \ref{ss_classes}.  This referee pointed out the relationship in equation (\ref{eq:T2inSig}) which led to the more general definition beneath.

\newpage

\section{Introduction} \label{SEC_Intro}

In this paper we present a new addition formula and sets of differential equations satisfied by the Abelian functions, with poles along the standard theta divisor, associated with a trigonal curve of genus four.  Perhaps of greater importance, is the discussion of how similar results may be derived for any suitable curve through the derivation of bases for the corresponding vector spaces.  The problem of choosing a basis of linearly independent functions from a wider set can be solved using the $\sigma$-function expansion, (see for example \cite{eemop07} and \cite{MEe09}).  The problem of identifying the sets of suitable functions from which to choose is the main focus of this paper.  As part of the discussion, we present a number of new explicit classes of such functions which may be used for any suitable curve, (i.e. without restriction on the degree of the curve variables or the genus).

We apply these new results to the trigonal curve of genus four, allowing us to derive a new addition formula in Theorem \ref{thm:35_3t2v}.  The existence of the addition formula was shown in \cite{bego08} but the formula could only now be derived after the explicit derivation of a basis for the vector space of Abelian functions having poles of order at most three along the standard theta divisor.  A similar addition formula has recently been derived for the trigonal curve of genus three in \cite{MEeo10}.  This also required the derivation of a new basis which used the general ideas discussed here.  Such formulae were also the topic of \cite{emo11} which presented results for the Weierstrass $\wp$ and $\sigma$-functions along with generalisations to genus two.

We also derive complete sets of differential equations for the cyclic trigonal curve of genus four.  This can be seen as a continuation of the work started in \cite{bego08}, however the complete sets presented here required both the knowledge of the bases and new efficient computational techniques first introduced in \cite{MEe09}.  Some of the relations are very lengthly and so not included in the paper itself.  However, they are stored in a variety of formats online at \cite{DBAFweb}.

\vspace{3mm}

Many of the results we present can be viewed as generalisations of classic results for elliptic functions.  We will conclude the
introduction below by reminding the reader of the relevant results from the theory of Weierstrass functions.  Then in Section \ref{SEC_Construct} we revise the definitions of the Abelian functions, discussing the general properties they satisfy.  We proceed in Section \ref{SEC_Bases} to consider the problem of determining bases for the vector spaces of such functions, presenting an explicit construction for the trigonal curve of genus four.  In Section \ref{SEC_DE} we derive sets of differential equations satisfied by the functions and in Section \ref{SEC_AF} we present the new addition formula.

\vspace{3mm}

The classical results for Weierstrass elliptic functions forms a template for our theory.  Let $\wp(u)$ be the \textit{Weierstrass   $\wp$-function}, which as an elliptic function has two complex periods $\omega_1,\omega_2$:
\begin{equation} \label{eq:Intro_period}
\wp(u + \omega_1) = \wp(u + \omega_2) = \wp(u), \qquad \mbox{for all } \, u \in \mathbb{C}.
\end{equation}
The $\wp$-function has the simplest possible pole structure for an elliptic function and satisfies a number of interesting properties.  For example, it can be used to parametrise an elliptic curve,
\begin{equation} \label{eq:Intro_ec}
y^2 = 4x^3 - g_2x - g_3,
\end{equation}
where $g_2$ and $g_3$ are constants.  It also satisfies the following well-known differential equations,
\begin{eqnarray}
\big(\wp'(u)\big)^2 &=& 4\wp(u)^3 - g_2\wp(u) - g_3,            \label{eq:Intro_elliptic_diff1}     \\
           \wp''(u) &=& 6\wp(u)^2 - \textstyle \frac{1}{2}g_2.  \label{eq:Intro_elliptic_diff2}
\end{eqnarray}
Weierstrass introduced an auxiliary function, $\sigma(u)$, in his theory which satisfies
\begin{eqnarray}
\wp(u) &=& - \frac{d^2}{d u^2} \log \big[ \sigma(u) \big].  \label{eq:Intro_elliptic_ps}
\end{eqnarray}
The $\sigma$-function plays a crucial role in both the generalisation and in applications of the theory.  It satisfies the following two term addition formula,
\begin{eqnarray}
- \frac{\sigma(u+v)\sigma(u-v)}{\sigma(u)^2\sigma(v)^2} = \wp(u) - \wp(v). \label{eq:Intro_elliptic_add}
\end{eqnarray}
In this document we present generalisations of equations (\ref{eq:Intro_period})$-$(\ref{eq:Intro_elliptic_add}) for higher genus functions.

\vspace{3mm}

Elliptic functions have been the subject of much study since their discovery and have been extensively used to enumerate solutions of non-linear wave equations.  Recent times have seen a revival of interest in the theory of Abelian functions, which have multiple independent periods, and so generalise the elliptic functions.  The periodicity property is usually defined in association with an underlying algebraic curve.  These functions have been shown to solve some of the differential equations arising in mathematical physics and have been used in a variety of applications.

\section{Constructing Abelian functions} \label{SEC_Construct}

In this paper we study Abelian functions associated with the  following set of algebraic curves.

\begin{definition} \label{def:HG_general_curves}
For two coprime integers $(n,s)$ with $s>n$ we define an \textbf{$\bm{(n,s)}$-curve} as a non-singular algebraic curve defined by $f(x,y)=0$, where
\begin{equation} \label{eq:general_curve}
f(x,y) = y^n + p_1(x)y^{n-1} + p_{2}(x)y^{n-2} + \cdots + p_{n-1}(x)y - p_{n}(x).
\end{equation}
Here $x$, $y$ are complex variables and $p_j(x)$ are polynomials in $x$ of degree (at most) $\lfloor{ js } / n \rfloor$.  We define a simple subclass of the curves by setting $p_j(x)=0$ for $0\leq j \leq n-1$.  These curves are defined by
\begin{equation} \label{eq:ct_curve}
f(x,y) = y^n - (x^s +\lambda_{s-1}x^{s-1}+\dots+\lambda_{1}x+\lambda_0)
\end{equation}
and are called \textbf{cyclic $\bm{(n,s)}$-curves}.  We follow previous publications and denote the curve constants by $\lambda_j$ for the cyclic curves and by $\mu_j$ for the general $(n,s)$-curves.  Note that in the literature the word \lq\lq cyclic" is sometimes replaced by \emph{\lq\lq strictly"} or \emph{\lq\lq purely (\(n\)-gonal)"}.
\end{definition}
Unless we specify otherwise, we will be working with the general $(n,s)$-curves.  In some cases we restrict to the cyclic case for either theoretical or computational reasons, as noted at the time.

We denote the surface defined by an $(n,s)$-curve by $C$.  The genus of $C$ is given by $g=\frac{1}{2}(n-1)(s-1)$ and the associated Abelian functions, to be defined shortly, will be multivariate with $g$ variables, $\bm{u} = (u_1, \dots, u_g)$.  As an example, the elliptic curve in equation (\ref{eq:Intro_ec}) is a (2,3)-curve and the associated Weierstrass $\sigma$ and $\wp$-functions depend upon a single complex variable $u$.

The $(n,s)$-curves with $n=2$ are generally defined to be \emph{hyperelliptic curves}, (when $s=3$ they reduce to elliptic curves).  Klein developed an approach to generalise the Weierstrass $\wp$-function to Abelian functions associated with hyperelliptic curves, as described in Baker's classic texts \cite{ba97} and \cite{ba07}.   This approach has motivated the general definitions in \cite{bel97} and \cite{eel00} of what we now call \emph{Kleinian $\wp$-functions}.  It is the properties of these and the generalised $\sigma$-function that are our objects of study.

In the last few years a good deal of progress has been made on the theory of Abelian functions associated to those $(n,s)$-curve with $n=3$, which we label \textit{trigonal curves}.  The functions associated to a trigonal curve were first studied in detail in \cite{bel00}, where the authors of \cite{bel97} furthered their methods for the hyperelliptic case to obtain realisations of the Jacobian variety and some key differential equations between the functions.  The simplest trigonal curve is the (3,4)-curve which is of genus three.  This was first studied in detail in \cite{eemop07} with new results recently presented in \cite{MEeo10}.  The other canonical class of trigonal curve is the (3,5)-curve, which has genus four and will be the example considered in this paper.  This class of curves is defined by
\begin{eqnarray}
& &y^3+(\mu_3x+\mu_5)y^2+(\mu_1x^3+\mu_4x^2+\mu_7x+\mu_{10})y \nonumber \\
& &\qquad = x^5+\mu_{3}x^4+\mu_{6}x^3+\mu_{9}x^2+\mu_{12}x+\mu_{15},
\label{eq:35}
\end{eqnarray}
with the cyclic restriction given by
\begin{eqnarray}
y^3 &= x^5 + \lambda_4x^4 + \lambda_3x^3 + \lambda_2x^2 + \lambda_1x + \lambda_0.
\label{eq:c35}
\end{eqnarray}
The Abelian functions associated with such curves were first studied in \cite{bg06} for use in an application to the Benney equations, shortly followed by the more detailed study of the curves given in \cite{bego08}.  The authors of these papers provided an explicit construction of the differentials on the curve, solved the Jacobi inversion problem and obtained a number of differential equations between the $\wp$-functions.  Some of the properties of higher genus trigonal curves have recently been explored in \cite{MEhgt10}.  The class of $(n,s)$-curves with $n=4$ are defined as \textit{tetragonal curves} and were considered for the first time in \cite{MEe09} and \cite{MEg09}.  The lowest genus tetragonal curve is of genus six, leading to some computational restrictions.

We define a set of weights for the theory, denoted by $\mathrm{wt}$ and often referred to as the {\it Sato weights}.  We start by setting
\begin{equation}  \label{eq:Sato_weight}
\mathrm{wt}(x)=-n, \ \ \mathrm{wt}(y)=-s
\end{equation}
and then choosing the weights of the curve parameters to be such that the curve equation is homogeneous.  We see that for cyclic curves this imposes $\mathrm{wt}( \lambda_j )=-n(s-j)$ while for the non-cyclic curves we usually label the $\mu_j$ with their weight.  For example, in equation (\ref{eq:35}) the $\mu_j$ all have weight $-j$ and the equation has weight $-15$ overall.  The weights of all other objects in the theory may be derived uniquely and will be commented on throughout the paper.  Note that all the objects in this paper have a definite Sato weight and all the equalities are homogeneous with respect to this weight.  A more detailed discussion of the weight properties may be found elsewhere, for example in \cite{MEe09}.

\vspace{3mm}

We now discuss how to construct the functions associated to the curve.  We start by choosing a basis for the space of differential forms of the first kind.  These are the differential 1-forms which are holomorphic on $C$.  There is a standard procedure to construct this basis for an $(n,s)$-curve, (see for example \cite{N10}).  Any (3,5)-curve has a basis given by
\begin{equation} \label{eq:du35}
\bm{du}=\Big(\frac{dx}{f_y(x,y)}, \frac{xdx}{f_y(x,y)}, \frac{ydx}{f_y(x,y)}, \frac{x^2dx}{f_y(x,y)}\Big).
\end{equation}
In general the basis is given by $h_idx/f_y, i=1 \dots g$ where the $h_i$ are monomials in $(x,y)$ whose structure may be predicted by the Weierstrass gap sequence.  We note here that some authors use a different normalisation when working with $(n,s)$-curves.  For example, in the hyperelliptic cases it is common for the coefficient of $x^s$ to be set to $4$ instead of $1$ in the curve definition.  This corresponds to the removal of a factor of $\tfrac{1}{2}$ from the basis of differentials, which would otherwise each have denominator $f_y=2y$.  It is simple to move between such different normalisations.

We next choose a symplectic basis in \(H_1(C,\mathbb{Z})\) of cycles (closed paths) upon the compact Riemann surface defined by $C$.  We denote these by
$\{\alpha_1, \dots \alpha_g,$ $\beta_1, \dots \beta_g\}$, with $g=4$ in the (3,5)-case.  We ensure the cycles have intersection numbers
\begin{eqnarray*}
\alpha_i \cdot \alpha_j = 0, \qquad \beta_i \cdot \beta_j = 0, \qquad
\alpha_i \cdot \beta_j = \delta_{ij} =
\left\{ \begin{array}{rl}
1 & \mbox{if }  i = j, \\
0 & \mbox{if }  i \neq j.
\end{array} \right.
\end{eqnarray*}
The choice of these cycles is not unique, but the functions we work with will be independent of the choice.

We introduce $\bm{dr}$ as a basis of differentials of the second kind.  These are meromorphic differentials on $C$ which have their only pole at $\infty$.  This basis is usually derived alongside a fundamental differential of the second kind which plays an important role in the theory of $\wp$-functions.  We do not explicitly use either of these objects in this paper and so refer the reader to \cite{N10} for the general theory and \cite{bg06} which contains all the details for the trigonal curve of genus four.

We can now define the standard period matrices associated to the curve as
\begin{eqnarray*}
\begin{array}{cc}
      \omega'  = \left( \oint_{\alpha_k} du_\ell \right)_{k,\ell = 1,\dots,g} &
\qquad\omega'' = \left( \oint_{ \beta_k} du_\ell \right)_{k,\ell = 1,\dots,g}  \\
        \eta'  = \left( \oint_{\alpha_k} dr_{\ell} \right)_{k,\ell = 1,\dots,g} &
\qquad  \eta'' = \left( \oint_{ \beta_k} dr_{\ell} \right)_{k,\ell = 1,\dots,g}
\end{array}.
\end{eqnarray*}
We define the period lattice $\Lambda$ formed from $\omega', \omega''$ by
\[
\Lambda = \big\{ \omega'\bm{m} + \omega''\bm{n}, \quad \bm{m},\bm{n} \in \mathbb{Z}^g \big\}.
\]
Note the comparison with equation (\ref{eq:Intro_period}) and that the period matrices play the role of the scalar periods in the elliptic case.  All the functions we treat in this paper are defined upon \(\mathbb{C}^g\) with coordinates usually expressed as
\begin{equation}
  \bm{u}=(u_1, \dots, u_g).
\end{equation}
Note that any point $\bu \in \mathbb{C}^g$ can be expressed as
\begin{eqnarray*}
\bu &= \sum_{i=1}^g \int_{\infty}^{P_i} \bm{du},
\end{eqnarray*}
where the $P_i$ are points upon $C$.  The period lattice $\Lambda$ is a lattice in the space $\mathbb{C}^g$.  Then the Jacobian variety of $C$ is presented by $\mathbb{C}^g/\Lambda$, and is denoted by $J$.
We define \(\kappa\) as the modulo \(\Lambda\) map,
\begin{equation}
\kappa \ : \ \mathbb{C}^g \to J.
\end{equation}
For $k=1$, $2$, $\dots$ define $\mathfrak{A}$, the \emph{Abel map} from the $k$-th symmetric product \(\mathrm{Sym}^k(C)\) of \(C\)  to $J$ by
\begin{eqnarray}
\mathfrak{A}: \mbox{Sym}^k(C) &\to&     J \nonumber \\
(P_1,\dots,P_k)   &\mapsto&
\left( \int_{\infty}^{P_1} \bm{du} + \dots + \int_{\infty}^{P_k} \bm{du} \right) \pmod{\Lambda},
\label{eq:Abel}
\end{eqnarray}
where the $P_i$ are again points upon $C$.  Denote the image of the $k$-th Abel map by $W^{[k]}$ and define the \emph{$k$-th standard theta subset} (often referred to as the $k$-th strata) by
\begin{eqnarray*}
\Theta^{[k]} = W^{[k]} \cup [-1]W^{[k]},
\end{eqnarray*}
where \([-1]\) means that
\begin{eqnarray*}
[-1](u_1, \dots ,u_g) = (-u_1, \dots ,-u_g).
\end{eqnarray*}
When $k=1$ the Abel map gives an embedding of the curve $C$ upon which we define $\xi= x^{-\frac{1}{n}}$ as the local parameter at the origin, $\mathfrak{A}_1(\infty)$.  We can then express $y$ and the basis of differentials using $\xi$ and integrate to give series expansions for $\bu$.  We can check the weights of $\bu$ from these expansions and see that they are prescribed by the Weierstrass gap sequence.  In particular, for the (3,5)-curve the variables $\bu=(u_1,\dots,u_4)$ have weights $(7,4,2,1)$.

\vspace{3mm}

\noindent We now consider functions that are periodic with respect to the lattice $\Lambda$.
\begin{definition} \label{def:HG_Abelian}
Let $\mathfrak{M}(\bu)$ be a meromorphic function of $\bu \in \mathbb{C}^g$.  Then $\mathfrak{M}$ is a \textbf{standard Abelian function associated with $\bm{C}$} if it has poles only along \(\kappa^{-1}(\Theta^{[g-1]})\) and satisfies, for all $\bm{\ell}\in \Lambda$,
\begin{equation} \label{eq:HG_Abelian}
\mathfrak{M}(\bu + \bm{\ell}) = \mathfrak{M}(\bu).
\end{equation}
\end{definition}

We define generalisations of the Weierstrass functions which satisfy equation (\ref{eq:HG_Abelian}), defined using the quasi-periodic function defined below.  First, let \(\bm{\delta} = \omega'\bm{\delta'}+\omega''\bm{\delta''}\) be the Riemann constant with base point $\infty$.  Then $[\bm{\delta}]$ is the theta characteristic representing the Riemann constant for the curve C with respect to the base point $\infty$ and generators $\{\alpha_j,\ \beta_j\}$ of $H_1(C,\mathbb{Z})$.  (See for example \cite{bel97} pp23-24.)

\begin{definition} \label{def:HG_sigma}
The \textbf{Kleinian $\bm{\sigma}$-function} associated to a general $(n,s)$-curve is defined using a multivariate $\theta$-function with characteristic \(\bm{\delta}\) as
\begin{eqnarray*}
\sigma(\bu) &=& c \exp \big( \textstyle \frac{1}{2} \bm{u} \eta' (\omega')^{-1} \bm{u}^T \big)
\cdot \theta[\bm{\delta}]\big((\omega')^{-1}\bm{u}^T \hspace*{0.05in} \big| \hspace*{0.05in} (\omega')^{-1} \omega''\big).  \\
&=& c \exp \big( \textstyle \frac{1}{2} \bm{u} \eta' (\omega')^{-1} \bm{u}^T \big)
\displaystyle \sum_{\bm{m} \in \mathbb{Z}^g} \exp \bigg[ 2\pi i \big\{  \\ & &
\textstyle \frac{1}{2} (\bm{m} + \bm{\delta'})^T (\omega')^{-1} \omega''(\bm{m}
+ \bm{\delta'}) + (\bm{m} + \bm{\delta'})^T ((\omega')^{-1} \bm{u}^T + \bm{\delta''} )\big\} \bigg].
\end{eqnarray*}
The constant $c$ is dependent upon the curve parameters and the basis of cycles and is fixed later,
following Lemma {\rm \ref{lem:sigexp}}.
\end{definition}
We now summarise the key properties of the $\sigma$-function.  See \cite{bel97} or \cite{N10} for the construction of the $\sigma$-function to satisfy these properties.
For any point $\bu \in \mathbb{C}^g$ we denote by $\bu'$ and $\bu''$ the vectors in $\mathbb{R}^g$ such that
\begin{equation}
\bu=\omega'\bm{u}'+\omega''\bm{u}''.
\end{equation}
Therefore a point  $\bm{\ell}\in\Lambda$ is written as
\begin{equation} \label{eq:HG_ell}
\bm{\ell} = \omega'\bm{\ell'} + \omega''\bm{\ell''} \in \Lambda, \qquad \bm{\ell'},\bm{\ell''} \in \mathbb{Z}^g.
\end{equation}
For $\bu, \bv \in \mathbb{C}^g$ and $\bm{\ell} \in \Lambda$, define $L(\bu,\bv)$ and $\chi(\bm{\ell})$ as follows:
\begin{eqnarray*}
L(\bu,\bv) &=& \bu^T \big( \eta'\bm{v'} + \eta''\bm{v''} \big), \\
\chi(\bm{\ell}) &=& \exp \big[ 2 \pi \mbox{i} \big\{  (\bm{\ell'})^T\delta'' - (\bm{\ell''})^T\delta'
+ \textstyle \frac{1}{2}(\bm{\ell'})^T \bm{\ell''} \big\} \big].
\end{eqnarray*}

\begin{lemma}
Consider the $\sigma$-function associated to an $(n,s)$-curve.
\begin{itemize}

\item It is an entire function on $\mathbb{C}^g$.

\item It has zeros of order one along the set  $\kappa^{-1}(\Theta^{[g-1]})$.  Further, we have $\sigma(\bu)~\neq~0$ outside the set.

\item For all $\bu \in \mathbb{C}^g, \bm{\ell} \in \Lambda$ the function has the quasi-periodicity property:
\begin{eqnarray} \label{eq:HG_quas}
\sigma(\bu + \bm{\ell}) = \chi(\bm{\ell})
\exp \left[ L \left( \bu + \frac{\bm{\ell}}{2}, \bm{\ell} \right) \right] \sigma(\bu).
\end{eqnarray}

\item It has definite parity given by
\begin{equation} \label{eq:sigparity}
\sigma(-\bu) = (-1)^{\frac{1}{24}(n^2-1)(s^2-1)}\sigma(\bu).
\end{equation}
Hence the function is even in the $(3,5)$-case.
\end{itemize}

\end{lemma}

\begin{proof}
The function is clearly entire from the definition, while the zeros and the quasi-periodicity are classical results, (see \cite{ba97}), that are fundamental to the definition of the function.  They both follow from the properties of the multivariate $\theta$-function.  The parity property is given by Proposition 4(iv) in \cite{N10}. \\
\end{proof}
We next define $\wp$-functions using an analogy of equation (\ref{eq:Intro_elliptic_ps}).  Since there is more than one variable we need to be clear which we differentiate with respect to.  We actually define multiple $\wp$-functions and introduce the following \emph{index notation}.

\begin{definition} \label{def:nip}
Define \textbf{$\bm{m}$-index Kleinian $\bm{\wp}$-functions} for $m\geq2$ by
\begin{eqnarray*}
\wp_{i_1,i_2,\dots,i_m}(\bu) = - \frac{\partial}{\partial u_{i_1}} \frac{\partial}{\partial u_{i_2}}\dots
\frac{\partial}{\partial u_{i_m}} \log \big[ \sigma(\bu) \big],
\end{eqnarray*}
where $i_1 \leq \dots \leq i_m \in \{1,\dots,g\}$.
\end{definition}
The $m$-index $\wp$-functions are meromorphic with poles of order $m$ when $\sigma(\bu)=0$.  We can check that they satisfy equation (\ref{eq:HG_Abelian}) and hence they are Abelian.   The $m$-index $\wp$-functions have definite parity with respect to the change of variables $\bu \to [-1]\bu$.  This is independent of the underlying curve, with the functions odd if $m$ is odd and even if $m$ is even.  Note that the ordering of the indices is irrelevant and so for simplicity we always order in ascending value.

When the $(n,s)$-curve is chosen to be the classic elliptic curve then the Kleinian $\sigma$-function coincides with the classic $\sigma$-function and the sole 2-index $\wp$-function coincides with the Weierstrass $\wp$-function.  The only difference is the notation with
\[
\wp_{11}(\bu) \equiv \wp(u), \quad \wp_{111}(\bu) \equiv \wp'(u), \quad \wp_{1111}(\bu) \equiv \wp''(u).
\]
Clearly, as the genus of the curve increases so do the number of associated $\wp$-functions.  In general, the number of $m$-index $\wp$-functions associated with a genus $g$ curve is
\[
\frac{ (g + m - 1)!}{m!(g - 1)!}.
\]
For the (3,5)-curve we study we have $g=4$ giving ten 2-index $\wp$-functions, twenty 3-index $\wp$-functions etc.  By considering Definition \ref{def:nip} we see that the weight of the $\wp$-functions is the negative of the sum of the weights of the variables indicated by the indices.  We note that curves of the same genus will, notationally, have the same $\wp$-functions.  However they exhibit different behaviour as indicated by the different weights.  When considering curves of genus four, we have the trigonal curve which we study and also a hyperelliptic curve, the (2,9)-case.  The weights of the associated 2-index $\wp$-functions are compared below.

\begin{center}
\begin{tabular}{|c|c|c|c|c|c|c|c|c|c|c|}\hline
\mbox{Curve} & $\wp_{11}$ & $\wp_{12}$ & $\wp_{13}$ & $\wp_{14}$ & $\wp_{22}$ & $\wp_{23}$ & $\wp_{24}$ & $\wp_{33}$ & $\wp_{34}$ & $\wp_{44}$ \\ \hline
\textbf{(2,9)}
& $-14$      & $-12$      &  $-10$     & $-8$       & $-10$      & $-8$       & $-6$       & $-6$       & $-4$       & $-2$       \\
\textbf{(3,5)}
& $-14$      & $-11$      &  $-9$      & $-8$       & $-8$       & $-6$       & $-5$       & $-4$       & $-3$       & $-2$       \\ \hline
\end{tabular}
\end{center}

We now introduce a final result detailing how the functions can be expressed using series expansions.

\begin{lemma} \label{lem:sigexp}
The Taylor series expansion of $\sigma(\bu)$ about the origin may be written as
\begin{equation} \label{eq:HG_SW}
\sigma(\bu) = K \cdot SW_{n,s}(\bu) + \sum_{k=0}^{\infty} C_{k}(\bu,\bm{\lambda}).
\end{equation}
Here $K$ is a constant, $SW_{n,s}$ the Schur-Weierstrass polynomial generated by $(n,s)$ and each $C_k$ a finite, polynomial composed of products of monomials in $\bm{u}$ of weight $k$ multiplied by monomials in the curve parameters of weight $-(\mbox{wt}(\sigma)-k)$.
\end{lemma}
\begin{proof}
We refer the reader to \cite{N10} for a proof of the relationship between the $\sigma$-function and the Schur-Weierstrass polynomials and note that this was first discussed in \cite{bel99}.  We see that the remainder of the expansion must depend on the curve parameters and split it up into the different $C_k$ using the weight properties.  We can see that each $C_k$ is finite since the number of possible terms with the prescribed weight properties if finite. \\
\end{proof}
Large expansions of this type were first introduced in \cite{bg06} in relation to the study of reductions of the Benney equations.  Since then they have been an integral tool in the investigation of Abelian functions.  Recently computational techniques based on the weight properties have been used to derive much larger expansions and we refer the reader to \cite{MEe09} and \cite{MEhgt10} for details of how to construct and use these expansions.  We note that such expansions are possible for the general curves, but that the calculations involved are far simpler for the cyclic cases.

The Schur-Weierstrass polynomials are Schur polynomials generated by a partition associated with $(n,s)$, (see \cite{bel99} for full details on the construction).  The Schur-Weierstrass polynomial generated by (3,5) is
\begin{equation}
SW_{3,5} = \textstyle \frac{1}{448}u_4^8 - \frac{1}{8}u_3^2u_{4}^4 - \frac{1}{4}u_{3}^4 + u_{2}^2 + u_{3}u_{4}^2u_{2} - u_{4}u_{1}. \label{eq:35SW}
\end{equation}
In Definition \ref{def:HG_sigma} we fix $c$ to be the value that makes the constant $K=1$ in the above lemma.  Some other authors working in this area may use a different constant and in general these choices are not equivalent.  However, the constant can be seen to cancel in the definition of the $\wp$-functions, leaving results between the functions independent of $c$.  Note that this choice of $c$ ensures that the Kleinian $\sigma$-function matches the Weierstrass $\sigma$-function when the $(n,s)$-curve is chosen to be the classic elliptic curve.

The connection with the Schur-Weierstrass polynomials also allows us to determine the weight of the $\sigma$-function as $(1/24)(n^2-1)(s^2-1)$.  In the (3,5)-case this gives $\sigma(\bu)$ weight $8$ .  The $\sigma$-function expansion has been calculated up to and including the polynomial $C_{50}$ in the (3,5)-case.  These expansions are stored online at \cite{DBAFweb}.

\section{Bases for the vector spaces of Abelian functions} \label{SEC_Bases}

In this section we discuss bases for the vector spaces of Abelian functions.  We first identify the vector spaces considered, then define some new generic Abelian functions before presenting some explicit results for the (3,5)-curve.  We note that everything in this section is valid for general $(n,s)$-curves rather than the cyclic restriction used in later sections.

\subsection{Classifying Abelian functions by pole structure} \label{ss_structure}

We can classify the Abelian functions according to their pole structure.  It is usual to denote by
$\Gamma \big( J, \mathcal{O}(m \Theta^{[k]} ) \big)$ the vector space of Abelian functions defined upon $J$ which have poles of order at most $m$, occurring only on the $k$th strata, $\Theta^{[k]}$.  The case where $k=g-1$ is of interest since all the Abelian functions we deal with have poles occurring here, on the $\Theta$-divisor.  Hence, for brevity, we just refer to such spaces as $\Gamma(m)$ throughout the rest of this paper.  A key problem in the theory of Abelian functions is the generation of bases for these vector spaces.  Note that the dimension of the space $\Gamma(m)$ is $m^g$ by the Riemann-Roch theorem for Abelian varieties, (see for example \cite{la82}).

Recall that the $m$-index $\wp$-functions all have poles of order $m$ and so are natural candidates for the construction of bases for these spaces.  In fact, for a genus one curve they are sufficient since the Weierstrass $\wp$-function and its derivatives can span all such spaces:
\begin{eqnarray*}
\Gamma(0) &=& \Gamma(1) = \mathbb{C} 1, \\
\Gamma(2) &=& \mathbb{C} 1 \,\oplus\, \mathbb{C} \wp, \\
\Gamma(3) &=& \mathbb{C} 1 \,\oplus\, \mathbb{C} \wp \,\oplus\ \mathbb{C} \wp^{\prime}, \\
\Gamma(4) &=& \mathbb{C} 1 \,\oplus\, \mathbb{C} \wp \,\oplus\ \mathbb{C} \wp^{\prime} \,\oplus\ \mathbb{C} \wp^{\prime\prime}, \\
 &\vdots&
\end{eqnarray*}
In general, the number of $\wp$-functions with poles of order at most $m$ is
\[
\sum_{N=2}^{m} \frac{ (g + N - 1)!}{N!(g - 1)!} = \frac{(g+m)!}{g!m!} - (g+1),
\]
which, as $g$ increases, will not grow as fast as the dimension of the space.  We hence need to identify a wider class of Abelian functions than the $\wp$-functions in order to construct such bases.  Note that the $\wp$-functions are not sufficient even for $g=2$.  In this case
\begin{eqnarray*}
\Gamma(0) &=& \Gamma(1) = \mathbb{C} 1, \\
\Gamma(2) &=& \mathbb{C} 1 \,\oplus\, \mathbb{C} \wp_{11} \,\oplus\, \mathbb{C} \wp_{12} \,\oplus\, \mathbb{C} \wp_{22},
\end{eqnarray*}
but an additional function must be included alongside the $\{\wp_{ijk}\}$ in the basis for $\Gamma(3)$.  This space has dimension 9 while there are only four 3-index $\wp$-functions to add.  The final entry is filled by the function $\wp_{11}\wp_{22} - \wp_{12}^2$.  Although each of these terms has poles of order four, together they cancel to leave poles of order three.  The bases for subsequent $\Gamma(m)$ in the $g=2$ case can be constructed from the entries in the basis for $\Gamma(m-1)$ together with the derivatives of those entries.  This may be verified by noting that the dimension of the spaces and the number of distinct function obtained through differentiation increase by the same amount each time.  It can also be checked that these function all have distinct weights and so must be linearly independent.

These cases of genus 1 and 2 curves, where the entries for the bases may be determined in general, are special.  They fall into the class where the theta divisor is non-singular and the $\mathcal{D}$-module structure of such cases is discussed in \cite{cn06}.  However, subsequent $(n,s)$-curves are not covered by this case and so new methods must be used to derive bases.  One of the authors of \cite{cn06} has conjectured the general hyperelliptic $\mathcal{D}$-module structure in \cite{N01s}, and explicit results have been derived for the genus $3$ and $4$ hyperelliptic curves in \cite{MEeo10}.  However, this paper discusses approaches relevant for all $(n,s)$-curves, including those that are non-hyperelliptic such as our example curve, the (3,5)-case.

\bigskip

Recall that an entire Abelian function must be a constant and that there is no Abelian function with a single pole of order one.  Hence those Abelian functions with poles of order two are the simplest and so are often referred to as \textit{fundamental Abelian functions}.  The basis problem has been solved in general for such functions, through the inclusion of the following extra class of Abelian functions.

\begin{definition} \label{def:Qdef}
We define an operator $\Upsilon_i$, now known as \textit{Hirota's bilinear operator} although it was used much earlier by Baker in \cite{ba07}.
\begin{equation} \label{eq:HBO}
\Upsilon_i = \frac{\partial}{\partial u_i} - \frac{\partial}{\partial v_i}.
\end{equation}
Then an alternative, equivalent definition of the 2-index $\wp$-functions is given by
\begin{equation} \label{eq:2ip_Delta}
\wp_{ij}(\bm{u}) = - \frac{1}{2\sigma(\bm{u})^2} \Upsilon_i\Upsilon_j \sigma(\bm{u}) \sigma(\bm{v}) \hspace*{0.1in} \Big|_{\bv=\bu}
\qquad i \leq j \in \{1,\dots,g\}.
\end{equation}
We extend this to define the \textbf{$\bm{m}$-index $\bm{Q}$-functions}, for $m$ even, by
\begin{equation} \label{eq:Qdef}
Q_{i_1, i_2,\dots,i_m}(\bm{u}) =  \frac{(-1)}{2\sigma(\bm{u})^2} \Upsilon_{i_1}\Upsilon_{i_2}\dots \Upsilon_{i_m} \sigma(\bm{u}) \sigma(\bm{v}) \hspace*{0.08in} \Big|_{\bv=\bu} \quad
\end{equation}
where $i_1 \leq \dots  \leq i_m \in \{1,\dots,g\}$.
\end{definition}

The $m$-index $Q$-functions are Abelian functions with poles of order two occurring when $\sigma(\bu)=0$.  Note that if you were to apply the definition with $m$ odd then the functions would be identically zero.  They are a generalisation of a function first used by Baker, referred to as the $Q$-function.  The generic 4-index $Q$-functions were introduced when research first started on the trigonal curves (and in the literature are also sometimes just referred to as \textit{$Q$-functions}).  The definition above was developed in \cite{MEe09} as increasing classes are required to deal with cases of higher genus, (see also \cite{MEhgt10}).  In this paper we only need to use 4-index $Q$-functions, which in \cite{eemop07} were shown to satisfy,
\begin{eqnarray} \label{eq:4iQ}
Q_{ijk\ell} = \wp_{ijk\ell} - 2 \wp_{ij}\wp_{k\ell}-2\wp_{ik}\wp_{j\ell} -2\wp_{i\ell}\wp_{jk}.
\end{eqnarray}
Similar expressions have been found for the higher index $Q$-functions, (see \cite{MEe09}).  Note that in some specific cases, alternatives to the $Q$-functions are available.  For example, in the (2,7)-case the basis for $\Gamma(2)$ requires one additional entry after including the $\wp$-functions.  The general approach would lead to the addition of a 4-index $Q$-function, $Q_{1333}$, but the function,
\begin{equation} \label{eq:DF}
\Delta=\wp_{11}\wp_{33}-\wp_{12}\wp_{23}-\wp_{13}^2+\wp_{13}\wp_{22}.
\end{equation}
may be used instead.  This $\Delta$-function, originally introduced by Baker in \cite{ba03}, has poles of order at most three in general.  (This may be checked using Definition \ref{def:nip}).  But in the case of the (2,7)-curve, these can be shown to cancel to leave poles of order two.

The use of $\Delta$ can be advantageous since it allows the theory to be completely realised in terms of 2 and 3-index $\wp$-functions, with all higher index $\wp$-functions given recursively in terms of these.  Similar functions have been found for the genus four hyperelliptic curve, (see \cite{MEeo10}), but they appear to be a feature unique to the hyperelliptic cases.  We note that the $\Delta$-function also arises naturally in the representation theoretic approach to hyperelliptic Abelian functions described in \cite{CA2008}.
In this paper, we concentrate on deriving classes of functions whose maximal pole order is known for any underlying curve, (such as the $Q$-functions), as opposed to searching for functions such as $\Delta$ which are only of use to problems on specific curves.  The approach presented here may hence be used in any applications which require the derivations of such bases.

\bigskip

When constructing a basis for $\Gamma(m)$ we can start by including the entries in the basis for $\Gamma(m-1)$.  We then know that the remaining entries have poles of order $m$.  A natural place to look for these functions is in the derivatives of the functions in $\Gamma(m-1)$.  Note that while the derivatives of $(m-1)$-index $\wp$-functions are $m$-index $\wp$-functions, the same is not true for the $Q$-functions.  For brevity we adopt the notation
\[
\partial_{m} Q_{ijkl}(\bu) = \frac{\partial}{\partial u_m} Q_{ijkl}(\bu)
\]
and similarly for other functions.

As discussed in \cite{N11}, the derivatives of existing basis functions will not be sufficient to find successive bases in the case where the theta divisor has singular points.  In Section 8 of \cite{eemop07} the authors introduced a new class of functions, $\wp^{[ij]}$, defined to be the $(i,j)$ minor of the matrix $[ \wp_{ij} ]_{3 \times 3}$.  The authors used these functions to overcome the problem for the three pole basis in the (3,4)-case and they appeared again for use in the (2,7)-case in \cite{MEeo10}.  We can check that these functions have poles of order at most three by substituting the $\wp$-functions for their definition.  This motivates a more general definition for similar functions, valid for curves of any genus.

\begin{definition} \label{def_DM}
Consider the matrix $[ \wp_{ij} ]_{g \times g}$.  Denote $\wp^{[i,j,k,l]}$ to be the determinant of the submatrix formed by rows $i$ and $j$ and columns $k$ and $l$.  (Note that we need $i \neq j$ and $k \neq l$).  So
\begin{eqnarray*}
\wp^{[i,j,k,l]} &=& \left| \begin{array}{cc} \wp_{ik} & \wp_{il} \\ \wp_{jk} & \wp_{jl} \end{array} \right| \\
&=& \wp_{ik}\wp_{jl} - \wp_{il}\wp_{jk}.
\end{eqnarray*}
We refer to these as \textbf{cross product} $\bm{\wp}$\textbf{-functions}.
\end{definition}
\noindent So the genus three $\wp^{[i,j]}$-functions used in \cite{eemop07} are the cross product functions in the genus three case.

Note that while each of the terms has poles of order four, together they cancel to leave poles of order three.  To see this substitute using Definition \ref{def:nip} so that
\[
\wp_{ij} = - \frac{\partial}{\partial u_{i}} \frac{\partial}{\partial u_{j}} \log \big[ \sigma(\bu) \big]
= \frac{\sigma_i\sigma_j - \sigma\sigma_{ij}}{\sigma^2}
\]
and so
\[
\wp^{[i,j,k,l]} = \frac{- \sigma_{ik}\sigma_j\sigma_l - \sigma_{jl}\sigma_i\sigma_k + \sigma\sigma_{ik}\sigma_{jl}
+ \sigma_{ij}\sigma_k\sigma_l + \sigma_{jl}\sigma_i\sigma_k - \sigma\sigma_{ij}\sigma_{kl}}{\sigma^3},
\]
from which it is clear that poles will be at most order three.  Although this class of $\wp^{[i,j,k,l]}$ functions completed the basis for functions with poles of order at most three in the (3,4)-case, it was not sufficient to complete the corresponding basis in the (2,7)-case.  These two genus three curves were recently discussed in detail in \cite{MEeo10}.  The difference between the bases resulted from the fact that two of the cross product functions were linearly dependent in the hyperelliptic case while they were independent in the trigonal case.  This is made possible through the different weight structures in the two cases.  In \cite{MEeo10} the final entry in the (2,7) basis for $\Gamma(3)$ was filled by the function
\[
T = \wp_{222}^2 + 2\wp_{22}^3 - \wp_{22}\wp_{2222},
\]
which although constructed from terms with poles of order six, has poles of order three overall.  This particular function belonged to a wider class of $\mathcal{T}$-functions derived through an attempt to match, in general, the poles of a quadratic term in the 3-index $\wp$-functions with a polynomial in fundamental Abelian functions.  These $\mathcal{T}$-functions have poles of order three for any $(n,s)$-curve.  The discovery of this class prompted the systematic study of how the $\wp$-functions may be combined to produce classes of functions with prescribed orders of poles, discussed in Section \ref{ss_classes}.

\vspace{3mm}

The classes of functions discussed in the next section are just the candidates for the basis entries.  It is of course still necessary to check which individual functions are linearly independent.  This may be achieved for particular cases using the series expansion for the $\sigma$-function.  The calculations involved can be lengthly and grow in CPU time and memory requirement with both the genus and the number of poles.  Significant computational simplifications can be made by writing procedures that take advantage of the weight structure present in the theory.  For example, when expanding the product of series it is only necessary to multiply those terms which will give the correct final weight, as the other terms must all cancel.  Many of the results in this paper required a new efficient code to evaluate arbitrary products of series at appropriate weights.  The basis in Theorem \ref{thm:35_3pole}, the addition formula in Theorem \ref{thm:35_3t2v} and the differential equations in Section \ref{SEC_DE} all made use of this and the code is available online at \cite{DBAFweb}.

We can further simplify the basis calculations by noting that there is a finite weight range that basis entries can take.  Entries in $\Gamma(m)$ can have weight no lower than $-m\mbox{wt}(\sigma)$, (see Lemma 3.4 in \cite{MEeo10}).  This stops us from testing those functions at a lower weight thus drastically reducing the amount of computation required to find such bases.  These lower weight functions must be expressible as a linear combination of the basis functions in which every term depends on the curve parameters.  We note that while these bases are not unique, (they can formulated using different functions), their weight structures are.  We conjecture that each basis must have an element with the minimal weight and note that the structure of these bases are still under investigation.

\subsection{Classes of functions for generic curves} \label{ss_classes}

The new classes of functions presented in this section were derived in turn by considering a prescribed order of poles and a set number of indices on the functions.  A general polynomial was constructed from $\wp$-functions in which each term has the set number of indices.  The arbitrary coefficients in this polynomial were then derived by substituting for Definition \ref{def:nip} and ensuring the poles of higher order are canceled.  If all coefficients need to be set to zero then no new class is discovered, however, there is usually a non-zero combination that gives the required pole order.

Since the basis problem for $\Gamma(2)$ may be solved in general by the $Q$-functions, we start by considering functions with poles of order three.  When restricted to two indices the only polynomial that may be constructed is a linear combination of 2-index $\wp$-functions which will have poles of order at most two.  Similarly, when restricted to three indices we can only form a linear combination of 3-index $\wp$-functions.  The first non trivial case occurs when we allow four indices.  We may construct polynomials from 4-index $\wp$-functions and the products of two 2-index $\wp$-functions.  There are two non-zero combinations for which the poles of order four cancel.  We identify these by considering an arbitrary sum of $\wp_{ijkl}$ and the different products of 2-index $\wp$-functions and imposing restrictions on the coefficients so that the poles of order four cancel. The two combinations turn out to be the cross-product functions from Definition \ref{def_DM} and the 2-index $Q$-functions from Definition \ref{def:Qdef}, (in which the poles of order three also cancel).

When we restrict to five indices then we may construct polynomials from 5-index $\wp$-functions and the products of a 2 and 3-index $\wp$-function.  There are two non-zero combinations for which the poles of order 5 and 4 cancel.  Again, these were identified by considering an arbitrary sum of $\wp_{ijklk}$ and the different products of 2 and 3-index $\wp$-functions and imposing restrictions on the coefficients so that the poles of order four and five cancel.  The first combination is given by
\begin{eqnarray}
\mathcal{B}_{ijklm} &=& \wp_{ij}\wp_{klm} + \textstyle \frac{1}{3}\big( \wp_{jk}\wp_{ilm} + \wp_{jl}\wp_{ikm} + \wp_{jm}\wp_{ikl} \nonumber \\
& &\qquad - 2\wp_{kl}\wp_{ijm} - 2\wp_{km}\wp_{ijl} - 2\wp_{lm}\wp_{ijk} \big), \label{eq:BF}
\end{eqnarray}
while the second can be shown to be a combination of $\mathcal{B}$-functions and derivatives of $4$-index $Q$-functions.  The label $\mathcal{B}$ is used as they are bilinear in the 2 and 3-index $\wp$-functions which gives them a special role in the construction of differential equations between the $\wp$-functions, discussed further in Section \ref{SEC_DE_bi}.  Note that the order of the indices in the $\wp$ and $Q$-functions is irrelevant and so we always write them in ascending order.  However, this is not the case for the $\mathcal{B}$-functions and so care must be taken when choosing indices.  (It is possible to replace the $\mathcal{B}$-function by four separate functions which only need be evaluated in ascending indices.)

We next allow six indices meaning we can build polynomials from the set of functions
\[
\{ \wp_{ijklmn}, \wp_{ijk}\wp_{lmn}, \wp_{ij}\wp_{klmn}, \wp_{ij}\wp_{kl}\wp_{mn} \}.
\]
There are three independent non-zero combinations in which the poles of order 6,5 and 4 cancel.  First there are the 6-index $Q$-functions as given by Definition \ref{def:Qdef}, (in which the poles of order 3 cancel also).  Secondly there is the class
\begin{eqnarray*}
&&\mathcal{T}_{ijklmn} = \textstyle \wp_{ijk}\wp_{lmn} + \frac{1}{3} \big(
4\wp_{ij}\wp_{kl}\wp_{mn} + 4\wp_{ij}\wp_{km}\wp_{ln} + 4\wp_{ij}\wp_{kn}\wp_{lm} \\
&&\quad + 4\wp_{ik}\wp_{jl}\wp_{mn} + 4\wp_{ik}\wp_{jm}\wp_{ln} + 4\wp_{ik}\wp_{jn}\wp_{lm}
+ 4\wp_{il}\wp_{jk}\wp_{mn}  \\
&&\quad + 4\wp_{im}\wp_{jk}\wp_{ln} + 4\wp_{in}\wp_{jk}\wp_{lm} - 5\wp_{il}\wp_{jm}\wp_{kn} - 5\wp_{il}\wp_{jn}\wp_{km} \\
&&\quad - 5\wp_{im}\wp_{jl}\wp_{kn} - 5\wp_{im}\wp_{jn}\wp_{kl} - 5\wp_{in}\wp_{jl}\wp_{km} - 5\wp_{in}\wp_{jm}\wp_{kl}   \\
&&\quad - 2\wp_{ij}\wp_{klmn} + 2\wp_{ik}\wp_{jlmn} - 2\wp_{jk}\wp_{ilmn} - 2\wp_{lm}\wp_{ijkn} - 2\wp_{ln}\wp_{ijkm}  \\
&&\quad - 2\wp_{mn}\wp_{ijkl} + \wp_{il}\wp_{jkmn} + \wp_{im}\wp_{jkln} + \wp_{in}\wp_{jklm}  + \wp_{jl}\wp_{ikmn}  \\
&&\quad + \wp_{jm}\wp_{ikln} + \wp_{jn}\wp_{iklm}  + \wp_{kl}\wp_{ijmn} + \wp_{km}\wp_{ijln} + \wp_{kn}\wp_{ijlm}
\big)
\end{eqnarray*}
which are the same as the $\mathcal{T}$-functions in \cite{MEeo10} with the $Q$-functions replaced by their expression in $\wp$-functions.  Finally, there is a class given by
\begin{eqnarray*}
&&\mathcal{S}_{ijklmn} = \wp_{ijklmn} - 3 \big( \wp_{ijk}\wp_{lmn} + \wp_{ijl}\wp_{kmn} + \wp_{ijm}\wp_{kln} + \wp_{ijn}\wp_{klm} \\
&& + \wp_{ikl}\wp_{jmn} + \wp_{ikm}\wp_{jln} + \wp_{ikn}\wp_{jlm} + \wp_{ilm}\wp_{jkn} + \wp_{iln}\wp_{jkm} + \wp_{imn}\wp_{jkl}  \big).
\end{eqnarray*}
This process may be continued indefinitely by increasing the number of indices allowed.  The process of determining coefficients so that higher order pole cancel is not very computationally intensive as it just involves deriving and solving sets of linear equations.  However, the process of checking whether new functions can be formed from derivatives of existing ones can become cumbersome.  The functions defined here, along with some others that were not used in this paper, can all be found online at \cite{DBAFweb}.  The only other class that is used explicitly in this paper is given below.
\begin{eqnarray*}
&&\mathcal{P}_{ijklmno} = \textstyle
- \wp_{ij}\wp_{kl}\wp_{mno} - \wp_{ij}\wp_{km}\wp_{lno}
+ \frac{1}{2}\wp_{ij}\wp_{kn}\wp_{lmo} - \wp_{ij}\wp_{lm}\wp_{kno} \\
&&\quad \textstyle + \frac{1}{2}\wp_{ij}\wp_{ln}\wp_{kmo} + \frac{3}{2}\wp_{ij}\wp_{no}\wp_{klm}
+ \frac{1}{2}\wp_{ij}\wp_{mn}\wp_{klo} - \wp_{ik}\wp_{jl}\wp_{mno} \\
&&\quad \textstyle - \wp_{ik}\wp_{jm}\wp_{lno} + \frac{1}{2}\wp_{ik}\wp_{jn}\wp_{lmo}
- \wp_{ik}\wp_{lm}\wp_{jno} + \frac{1}{2}\wp_{ik}\wp_{ln}\wp_{jmo} \\
&&\quad \textstyle + \frac{1}{2}\wp_{ik}\wp_{mn}\wp_{jlo} + \frac{3}{2}\wp_{ik}\wp_{no}\wp_{jlm}
- \wp_{il}\wp_{jk}\wp_{mno} - \wp_{il}\wp_{jm}\wp_{kno} \\
&&\quad \textstyle + \frac{1}{2}\wp_{il}\wp_{jn}\wp_{kmo} - \wp_{il}\wp_{km}\wp_{jno}
+ \frac{1}{2}\wp_{il}\wp_{kn}\wp_{jmo} + \frac{1}{2}\wp_{il}\wp_{mn}\wp_{jko} \\
&&\quad \textstyle + \frac{3}{2}\wp_{il}\wp_{no}\wp_{jkm} - \wp_{im}\wp_{jk}\wp_{lno}
- \wp_{im}\wp_{jl}\wp_{kno} + \frac{1}{2}\wp_{im}\wp_{jn}\wp_{klo} \\
&&\quad \textstyle - \wp_{im}\wp_{kl}\wp_{jno} + \frac{1}{2}\wp_{im}\wp_{kn}\wp_{jlo}
+ \frac{1}{2}\wp_{im}\wp_{ln}\wp_{jko} + \frac{3}{2}\wp_{im}\wp_{no}\wp_{jkl} \\
&&\quad \textstyle - \wp_{in}\wp_{jk}\wp_{lmo} - \wp_{in}\wp_{jl}\wp_{kmo}
- \wp_{in}\wp_{jm}\wp_{klo} - \frac{3}{2}\wp_{in}\wp_{jo}\wp_{klm} \\
&&\quad \textstyle - \wp_{in}\wp_{kl}\wp_{jmo}- \wp_{in}\wp_{km}\wp_{jlo}
- \frac{3}{2}\wp_{in}\wp_{ko}\wp_{jlm} - \wp_{in}\wp_{lm}\wp_{jko} \\
&&\quad \textstyle - \frac{3}{2}\wp_{in}\wp_{lo}\wp_{jkm} - \frac{3}{2}\wp_{in}\wp_{mo}\wp_{jkl}
+ \wp_{jk}\wp_{lo}\wp_{imn} + \wp_{jk}\wp_{mo}\wp_{iln} \\
&&\quad \textstyle + \wp_{jk}\wp_{no}\wp_{ilm} + \wp_{jl}\wp_{ko}\wp_{imn}
+ \wp_{jl}\wp_{mo}\wp_{ikn} + \wp_{jl}\wp_{no}\wp_{ikm} \\
&&\quad \textstyle + \wp_{jm}\wp_{ko}\wp_{iln} + \wp_{jm}\wp_{no}\wp_{ikl}
+ \wp_{jm}\wp_{lo}\wp_{ikn} - \frac{1}{2}\wp_{jn}\wp_{ko}\wp_{ilm} \\
&&\quad \textstyle - \frac{1}{2}\wp_{jn}\wp_{lo}\wp_{ikm} - \frac{1}{2}\wp_{jn}\wp_{mo}\wp_{ikl}
+ \wp_{jo}\wp_{kl}\wp_{imn} + \wp_{jo}\wp_{km}\wp_{iln} \\
&&\quad \textstyle - \frac{1}{2}\wp_{jo}\wp_{kn}\wp_{ilm} + \wp_{jo}\wp_{lm}\wp_{ikn}
- \frac{1}{2}\wp_{jo}\wp_{ln}\wp_{ikm} - \frac{1}{2}\wp_{jo}\wp_{mn}\wp_{ikl} \\
&&\quad \textstyle + \wp_{lm}\wp_{no}\wp_{ijk} - \frac{1}{2}\wp_{ln}\wp_{mo}\wp_{ijk}
- \frac{1}{2}\wp_{lo}\wp_{mn}\wp_{ijk} + \wp_{kl}\wp_{mo}\wp_{ijn} \\
&&\quad \textstyle + \wp_{kl}\wp_{no}\wp_{ijm} + \wp_{km}\wp_{lo}\wp_{ijn}
+ \wp_{km}\wp_{no}\wp_{ijl} - \frac{1}{2}\wp_{kn}\wp_{lo}\wp_{ijm} \\
&&\quad \textstyle - \frac{1}{2}\wp_{kn}\wp_{mo}\wp_{ijl} + \wp_{ko}\wp_{lm}\wp_{ijn}
- \frac{1}{2}\wp_{ko}\wp_{ln}\wp_{ijm} - \frac{1}{2}\wp_{ko}\wp_{mn}\wp_{ijl} \\
&&\quad \textstyle + \frac{1}{2} \big(
\wp_{mno}\wp_{ijkl}
- \wp_{ijn}\wp_{klmo} + \wp_{jno}\wp_{iklm}
+ \wp_{lno}\wp_{ijkm} - \wp_{ikn}\wp_{jlmo} \\
&&\quad \textstyle - \wp_{iln}\wp_{jkmo} - \wp_{imn}\wp_{jklo} + \wp_{kno}\wp_{ijlm}
\big)
+ \wp_{in}\wp_{jklmo} - \wp_{no}\wp_{ijklm}
\end{eqnarray*}
These $\mathcal{P}$-functions have poles of order three for any $(n,s)$-curve and use seven indices.  The $\mathcal{T}, \mathcal{S}$ and $\mathcal{P}$ labels were not chosen to signify anything of importance.  We can of course repeat this procedure to find functions with higher orders of poles.  Appendix \ref{APP_Pole4} contains details of classes of functions which have poles of order four for any $(n,s)$-curve.  Although these functions are not used in the remainder of this paper, they have been used in deriving bases for genus three curves \cite{MEeo10}.

\subsubsection*{Formulation as operators on sigma-functions}

The Hirota operator, introduced in equation (\ref{eq:HBO}), can be expressed as a determinant,
\[
\Upsilon_i = \left| \begin{array}{cc}
1 & 1 \\
\frac{\partial}{\partial v_i} & \frac{\partial}{\partial u_i}
\end{array} \right|.
\]
Then one possible generalisation is given by
\[
\Upsilon_{iii} = \left| \begin{array}{ccc}
1 & 1 & 1 \\
\frac{\partial}{\partial u_i} & \frac{\partial}{\partial v_i} & \frac{\partial}{\partial w_i} \\
\frac{\partial^2}{\partial u_i^2} & \frac{\partial^2}{\partial v_i^2} & \frac{\partial^2}{\partial w_i^2}
\end{array} \right|
\]
and the functions
\begin{equation} \label{eq:Xifunct}
\Xi_{i,j} = \frac{1}{\sigma(u)^3}\Upsilon_{iii}\Upsilon_{jjj}\sigma(\bm{u})\sigma(\bm{v})\sigma(\bm{w}) \Big|_{\bm{w}=\bm{v}=\bm{u}} .
\end{equation}
are Abelian with poles of order three and so belong to $\Gamma(3)$.  Note the similarity with the definition of the $Q$-functions in Definition \ref{def:Qdef}.  Although this new class of functions is not equivalent to any of the others we introduced above, we do have that
\begin{equation} \label{eq:T2inSig}
T_{222222}(\bm{u}) = \frac{1}{\sigma(u)^3}\Upsilon_{222}\Upsilon_{222}\sigma(\bm{u})\sigma(\bm{v})\sigma(\bm{w}) \Big|_{\bm{w}=\bm{v}=\bm{u}} .
\end{equation}
By searching amongst terms with appropriate indices we find that the class can be expressed using $\wp$-functions as
\[
\Xi_{i,j} = 6 \big( \wp_{ij}\wp_{iijj} - \wp_{iij}\wp_{ijj} - 2\wp_{ij}^3 \big) = 6\mathcal{T}_{iijijj}.
\]
Hence these $\Xi$-functions are a subclass of the $\mathcal{T}$-functions introduced above.  It is anticipated that there may be a more general definition that encompasses the entire class of $\mathcal{T}$-functions and indeed that all the functions introduced in this section may have alternative definitions involving operators acting on $\sigma$-functions.  However, it is not clear yet whether these alternative definitions offer any advantage other than brevity of notation.  For example, further functions in $\Gamma(3)$ may be defined by applying further $\Upsilon_{iii}$ operators to the definition in equation (\ref{eq:Xifunct}).  However, the resulting functions are a lengthly polynomial in $\sigma$-derivatives involving derivatives of order 12 and it is computationally disadvantageous to use these in place of the other functions introduced.  The derivation of Abelian functions using this approach is still a topic under investigation.

\subsection{The trigonal curve of genus four}

In this section we use the generic results of the previous section to derive a basis for our example case, the (3,5)-curve.  We treat the most general \((3,5)\)-curve, namely the curve defined by equation (\ref{eq:35}).

Lemma 5.1 in \cite{bego08} identified a basis for $\Gamma(2)$ in the (3,5)-case as
\begin{eqnarray}
&&\{ 1, \wp_{11}, \wp_{12}, \wp_{13}, \wp_{14}, \wp_{22}, \wp_{23}, \wp_{24}, \wp_{33}, \wp_{34}, \wp_{44}, \nonumber \\
&&\qquad Q_{1144}, Q_{1244}, Q_{2233}, Q_{1444}, Q_{2444} \}. \label{eq:35_2pole}
\end{eqnarray}

\begin{theorem} \label{thm:35_3pole}
A basis for $\Gamma(3)$ associated with the (3,5)-curve is
\[
\left\{ \begin{array}{ccccccccccccccccc}
(\ref{eq:35_2pole}),
& \wp_{111}, & \wp_{222}, & \partial_{4}Q_{2444}, & \partial_{1}Q_{2233}, & \wp^{[3434]}, & \wp^{[1423]}, & \mathcal{T}_{133344}, \\
& \wp_{112}, & \wp_{223}, & \partial_{3}Q_{2444}, & \partial_{4}Q_{1244}, & \wp^{[2434]}, & \wp^{[1323]}, & \mathcal{T}_{114444}, \\
& \wp_{113}, & \wp_{224}, & \partial_{2}Q_{2444}, & \partial_{3}Q_{1244}, & \wp^{[2334]}, & \wp^{[1414]}, & \mathcal{T}_{222233}, \\
& \wp_{114}, & \wp_{233}, & \partial_{1}Q_{2444}, & \partial_{2}Q_{1244}, & \wp^{[2424]}, & \wp^{[1224]}, & \mathcal{T}_{222224}, \\
& \wp_{122}, & \wp_{234}, & \partial_{4}Q_{1444}, & \partial_{1}Q_{1244}, & \wp^{[2324]}, & \wp^{[1314]}, & \mathcal{T}_{222222}, \\
& \wp_{123}, & \wp_{244}, & \partial_{3}Q_{1444}, & \partial_{4}Q_{1144}, & \wp^{[1434]}, & \wp^{[1223]}, & \mathcal{P}_{2222244} \\
& \wp_{124}, & \wp_{333}, & \partial_{2}Q_{1444}, & \partial_{3}Q_{1144}, & \wp^{[1334]}, & \wp^{[1313]}, & \\
& \wp_{133}, & \wp_{334}, & \partial_{1}Q_{1444}, & \partial_{2}Q_{1144}, & \wp^{[2323]}, & \wp^{[1214]}, & \\
& \wp_{134}, & \wp_{344}, & \partial_{4}Q_{2233}, & \partial_{1}Q_{1144}, & \wp^{[1424]}, & \wp^{[1213]}, & \\
& \wp_{144}, & \wp_{444}, & \partial_{2}Q_{2233}, &                       & \wp^{[1324]}, & \wp^{[1212]}, &
\end{array}
\, \right\}.
\]
Here (\ref{eq:35_2pole}) refers to the 16 elements in the basis for $\Gamma(2)$ presented above and $\partial_{i}f$ refers to the derivative with respect to $u_i$ of the function $f$.  The functions $\wp^{[ijkl]}, \mathcal{T}_{ijklmn}, \mathcal{P}_{ijklmno}$ were all defined previously in this section.
\end{theorem}

\begin{proof}
The dimension of the space is $3^g=3^4=81$ by the Riemann-Roch theorem for Abelian varieties.  All the selected elements belong to the space as they are either elements with poles of order two, derivatives of these elements or functions explicitly constructed to be in this space.  We can easily check their linear independence using the $\sigma$-function expansion.  (Maple is helpful for such a computation).

To actually construct the basis we started by including the 16 functions from basis (\ref{eq:35_2pole}) for the functions with poles of order at most two.  We then know that the remaining entries must have poles of order three.  We start by looking for entries from the set of derivatives of the basis
(\ref{eq:35_2pole}).  Note that these functions need not be linearly independent.  In fact in this case we see that four of the $Q$-derivatives have weight $-14$;
\[
\partial_1Q_{2444}, \, \partial_2Q_{1444}, \, \partial_3Q_{2233}, \, \partial_4Q_{1244},
\]
but that one of these may be expressed using the other three.  We are free to choose which one of the four we omit from the basis.  We test at decreasing weight levels and look to see whether these functions can be written as a linear combination upon substitution of the series expansions.  (We use Maple for these calculations and a more detailed discussion of such calculations is given in \cite{MEe09}.)

Note that while this theorem holds for the general (3,5)-curve in equation (\ref{eq:35}), we need only use the series expansions associated to the cyclic (3,5)-curve in equation (\ref{eq:c35}).  This is because if an element cannot be expressed using the basis with the restriction on the parameters, then neither will it be expressible with the wider set of parameters.  Further, we only need to use sufficient expansion to give non-zero evaluations of the functions considered in order to check whether they are linearly independent.

After examining all these functions we find that 55 basis elements have been identified.  The next class of functions to be examined were the $\wp^{[ijkl]}$-functions given in Definition \ref{def_DM}.  We have to test whether they are linearly independent both with themselves and with the functions already added to the basis.  Another simplification that may be made to such calculations is to note that all these functions have definite parity according to the number of indices.  That is, the functions with an odd number of indices are odd and those with an even number of indices are even.  This can be concluded from the parity properties of the $\wp$-functions and that all the functions here are constructed from these.  So for example, when testing to see whether the $\wp^{[ijkl]}$ may be added to the basis we need only see whether they are linearly independent of the entries of $\Gamma(2)$, since all those functions were even and their derivatives odd.  We identify a further 20 basis entries from the $\wp^{[ijkl]}$-functions.  Unlike the (3,4)-case, these functions were not sufficient to complete the basis.

To find the final functions we considered some of the new generic functions defined in the previous section.  First the $\mathcal{B}$-functions were examined, but in this case they were all linearly dependent on the existing basis functions.  We note that it is possible to reformulate $\Gamma(3)$ using 19 of the $\mathcal{B}$-functions in place of the 19 $Q$-derivatives.  This formulation may seem less logical but it has some advantages as discussed in Section \ref{SEC_DE_bi}.
To fill the remaining six basis entries it is necessary to consider functions with more indices.  First the $\mathcal{T}$-functions were considered and five more entries identified.  We note that the $\mathcal{T}$-functions chosen in the basis are not unique.  For example, at weight $-15$ there are 10 possible distinct $\mathcal{T}$-functions and any one of them may be chosen as the basis entry.  We choose the function $\mathcal{T}_{133344}$ at random.

The final entry was then filled by one of the $\mathcal{P}$-functions of weight $-22$.  Again, this function was just one of a group of 7-index functions, any one of which would have been satisfactory.  We note that it would have not been possible to replace this with a 6-index function however, since such a function would be even while this final basis function is odd.  This is why none of the $\mathcal{S}$-functions were required in addition to the $\mathcal{T}$-functions. \\
\end{proof}

\section{Differential Equations} \label{SEC_DE}

The Kleinian $\wp$-functions satisfy a variety of differential equations which we review in this section.  Some of these differential equations can be found occurring naturally in areas of mathematical physics, so the $\wp$-functions can be used to give solutions to a variety of important problems.  In this section we consider three main classes of differential equations and present complete explicit sets for the functions associated with the (3,5)-curve.  We present the relations associated with the cyclic restriction of the curve in equation (\ref{eq:c35}).  Relations for the general curves can be derived in a similar fashion at a greater computational cost.  The sets of differential equations in this paper are presented in decreasing weight order as indicated by the bold number in brackets.

\subsection{Four-index relations} \label{SEC_DE_4}

We first consider a set of relations to express the 4-index $\wp$-functions.  These \textit{4-index relations} are the generalisation of equation (\ref{eq:Intro_elliptic_diff2}) from the elliptic case.  We aim to express each 4-index $\wp$-function as a degree two polynomial in the 2-index $\wp$-functions, in comparison with equation (\ref{eq:Intro_elliptic_diff2}).  Then, through differentiation and manipulation of this set, we could express all higher index $\wp$-functions as polynomials in the 2 and 3-index $\wp$-functions.  Examples of such sets are given for the (2,5)-case by Baker in \cite{ba07} and for the (2,7)-case by Buchstaber, Enolskii and Leykin in \cite{bel97}.

There are various ways to derive such relations, with a recent survey given in \cite{eeg10}.  A constructive way is to consider the basis for $\Gamma(2)$ and the set of 4-index Q-functions.  Each Q-function has poles of order at most two and so belongs to the vector space $\Gamma(2)$.  Hence each Q-function can be expressed as a linear combination of basis entries.  (The explicit linear combination can be identified using the $\sigma$-expansion as discussed in \cite{MEe09}).

In the hyperelliptic cases the bases for $\Gamma(2)$ can be completed using 2-index $\wp$-functions and quadratic polynomials constructed of these, (such as the $\Delta$-function discussed in Section \ref{ss_structure}).  So each of the 4-index Q-functions, and hence 4-index $\wp$-functions, can be expressed as desired giving a set of 4-index relations.  However, in non-hyperelliptic cases the basis for $\Gamma(2)$ can not be given purely in 2-index $\wp$-functions and so such a set is unobtainable.  As discussed in the previous section, we can always complete the basis using the $Q$-functions.  However, this means that some 4-index $Q$-functions will be in the basis and hence some 4-index $\wp$-functions can not be expressed as desired.

The ability to construct the 4-index relations using only 2-index $\wp$-functions appears to be a feature unique to the hyperelliptic cases.  A more appropriate definition for \textit{4-index relations} seems to be a set that expresses all the 4-index $\wp$-functions using a degree two polynomial in the fundamental basis functions.  (Note that although the polynomials are of degree 2, the Q-functions used will only need to appear linearly.)

\begin{theorem}
The 4-index relations for the functions associated with the cyclic (3,5)-curve are
\begin{eqnarray}
\bm{(-4)} \quad \wp_{4444} &=& \textstyle  - 3\wp_{33} + 6\wp_{44}^{2} \label{eq:35_p4444}
\\
\bm{(-5)} \quad \wp_{3444} &=& \textstyle  3\wp_{24} + 6\wp_{34}\wp_{44} \label{eq:35_p3444}
\\
\bm{(-6)} \quad \wp_{3344} &=& \textstyle  - \wp_{23} + 2\wp_{34}\lambda_{4} + 2\wp_{33}\wp_{44} + 4 \wp_{34}^{2}
\nonumber \\
\bm{(-7)} \quad \wp_{3334} &=& \textstyle  - Q_{{2444}}+ 6\wp_{33}\wp_{34}
\nonumber \\
\bm{(-7)} \quad \wp_{2444} &=& \textstyle  Q_{{2444}}+ 6\wp_{24}\wp_{44}
\nonumber \\
\bm{(-8)} \quad \wp_{2344} &=& \textstyle  4\wp_{14} - \wp_{22} + 2\wp_{24}\lambda_{4}+ 2\wp_{23}\wp_{44} + 4\wp_{24}\wp_{34}
\nonumber
\end{eqnarray}
\begin{eqnarray}
\bm{(-8)} \quad \wp_{3333} &=& \textstyle  12\wp_{14} - 3\wp_{22} + 6 \wp_{33}^{2}
\nonumber \\
&\vdots& \nonumber \\
\bm{(-25)} \quad \wp_{1112} &=& \textstyle 6\wp_{11}\wp_{12} + \big(2\lambda_{0} + 2\lambda_{4}\lambda_{1}\big)Q_{{1444}} + 6\wp_{33}\lambda_{3}\lambda_{0}
\nonumber \\ & & - 3Q_{{1244}}\lambda_{1}
\nonumber \\
\bm{(-28)} \quad \wp_{1111} &=& \textstyle 6 \wp_{11}^{2} - 3\wp_{33} {\lambda_{1}}^{2} + 8\lambda_{4}\lambda_{0}Q_{{1444}}
+ 12\wp_{33}\lambda_{2}\lambda_{0} \nonumber \\ & & - 12Q_{{1244}}\lambda_{0} \nonumber
\end{eqnarray}
with the full set given in Appendix \ref{APP_4index}.  The relations up to weight $-19$ were presented originally in \cite{bego08}, but a full set has not been available until now.
\end{theorem}
\begin{proof}
As discussed above, these follow from the basis for $\Gamma(2)$ given in equation (\ref{eq:35_2pole}) above.  The five $Q$-functions used in the basis appear linearly in the equations.  The coefficients were determined using the $\sigma$-expansion, with more details on such calculations available in \cite{MEe09} for example. \\
\end{proof}

\subsection{Quadratic relations} \label{SEC_DE_quad}

It is natural to next consider a set of differential equations to generalise equation (\ref{eq:Intro_elliptic_diff1}). Such a set should give expressions for the product of two 3-index $\wp$-functions and so we refer to them as \textit{quadratic 3-index relations}.  The natural generalisation would express each product as a degree three polynomial in 2-index $\wp$-functions.  Examples of such generalisations have been found in hyperelliptic cases, first for the genus two curve in \cite{ba07} and more recently for the genus 3 case.  This is a (2,7)-curve and the corresponding relations were first considered in \cite{bel97} with a complete set recently derived in \cite{MEeo10}.

However, as with the 4-index relations, this generalisation is only possible for the hyperelliptic cases.  We have explicitly checked that such relations do not exist in a variety of non-hyperelliptic cases and so propose the modified definition of \textit{quadratic 3-index relations} to be a set of differential equations that expresses all the products of 3-index $\wp$-functions using a degree three polynomial in the fundamental basis functions.  We derive such a set of equations for the cyclic (3,5)-case.

\begin{theorem} \label{thm:Quad_35set}
Each of the 210 products of pairs of 3-index $\wp$-functions associated with the cyclic (3,5)-curve may be expressed as a degree three polynomial in the fundamental basis functions.  Some of the relations are given below with the full set available online at \cite{DBAFweb}.  (Note that the polynomial is of degree three but the $Q$-functions need only appear quadratically, or multiplied by a $\wp$-function.)
\begin{eqnarray*}
&&\bm{(-6)} \quad \hspace*{0.285in} \wp_{444}^2 = \textstyle 4\wp_{44}^3 - 4\wp_{23} - 4\wp_{33}\wp_{44}
+ \wp_{34}^2 + 2\wp_{34}\lambda_{4} + \lambda_{4}^2 - 4\lambda_{3}
\\
&&\bm{(-7)} \quad \wp_{344}\wp_{444} = \textstyle 4\wp_{34}\wp_{44}^2 + 2\wp_{24}\wp_{44}
- \wp_{33}\wp_{34} - \lambda_{4}\wp_{33} - \frac{2}{3}Q_{2444} \\
&&\bm{(-8)} \quad \hspace*{0.285in} \wp_{344}^2 = \textstyle 4\wp_{34}^2\wp_{44} + 4\wp_{14}
+ 4\wp_{24}\wp_{34} + \wp_{33}^2
\\
&&\bm{(-8)} \quad \wp_{334}\wp_{444} = \textstyle 2\wp_{34}^2\wp_{44} + 2\wp_{33}\wp_{44}^2
- 4\wp_{14} + 2\wp_{22} - 2\wp_{23}\wp_{44} \\
&&\quad \textstyle - \wp_{24}\wp_{34}  - 2\wp_{33}^2 - \wp_{24}\lambda_{4} + 2\lambda_{4}\wp_{34}\wp_{44}
\end{eqnarray*}
\begin{eqnarray*}
&&\bm{(-9)} \quad \wp_{334}\wp_{344} = \textstyle  2\big( \wp_{34}^3 + \wp_{44}\wp_{33}\wp_{34} - \wp_{13}
- \wp_{23}\wp_{34} + \lambda_{4}\wp_{34}^2 \big) + \wp_{24}\wp_{33}
\\
&&\bm{(-9)} \quad \wp_{333}\wp_{444} = \textstyle 6\wp_{44}\wp_{33}\wp_{34} - 4\lambda_{4}\wp_{34}^2
+ 2\wp_{13} + 7\wp_{23}\wp_{34} + 2\wp_{24}\wp_{33}  \\
&&\quad \textstyle + \lambda_{4}\wp_{23} - 2\lambda_{4}\wp_{33}\wp_{44} - 2\wp_{34}^3
- 2\wp_{34}\lambda_{4}^2 + 6\wp_{34}\lambda_{3} + 2\lambda_{2} - 2\wp_{44}Q_{2444}
\\
&&\bm{(-9)} \quad \wp_{244}\wp_{444} = \textstyle - 2\wp_{13} + \wp_{23}\wp_{34}
- 2\wp_{24}\wp_{33} - \lambda_{4}\wp_{23} + 2\wp_{34}\lambda_{3} - 2\lambda_{2} \\
&&\quad \textstyle + 4\wp_{24}\wp_{44}^2 + \frac{2}{3}\wp_{44}Q_{2444}  \\
&& \hspace*{1.2in} \vdots \\
&&\bm{(-37)} \quad \wp_{111}\wp_{113} = \textstyle 4\wp_{11}^2\wp_{13}
- \frac{2}{3}\lambda_{1}\lambda_{0}Q_{1444}
- \frac{4}{3}\lambda_{4}\lambda_{1}^2Q_{1444}
- 4\lambda_{0}\wp_{13}Q_{1244} \\
&&\quad \textstyle + \frac{8}{3}\lambda_{4}\lambda_{0}\wp_{13}Q_{1444}
+ 4\lambda_{4}\lambda_{2}\lambda_{0}Q_{1444}
+ 2\lambda_{0}\wp_{12}^2
+ 2\lambda_{0}^2Q_{2444}
+ 2\lambda_{1}\wp_{11}\wp_{12} \\
&&\quad \textstyle - \lambda_{1}^2\wp_{13}\wp_{33}
- \lambda_{2}\lambda_{1}^2\wp_{33}
- 3\lambda_{3}\lambda_{1}\lambda_{0}\wp_{33}
+ 4\lambda_{2}\lambda_{0}\wp_{13}\wp_{33}
+ 2\lambda_{1}^2Q_{1244} \\
&&\quad \textstyle - 6\lambda_{0}\wp_{11}\wp_{22}
+ 4\lambda_{0}\wp_{11}\wp_{14}
+ 4\lambda_{2}^2\lambda_{0}\wp_{33}
- 6\lambda_{2}\lambda_{0}Q_{1244}
+ 8\lambda_{4}\lambda_{0}^2\wp_{33}
\\
&&\bm{(-39)} \quad \wp_{111}\wp_{112} = \textstyle 4\wp_{11}^2\wp_{12}
- \frac{2}{3}\lambda_{0}\wp_{11}\wp_{2333}
- 2\lambda_{3}\lambda_{0}\wp_{13}^2
- 6\lambda_{3}\lambda_{0}^2\wp_{34}\\
&&\quad \textstyle - 2\lambda_{1}^2\wp_{12}\wp_{33}
+ \frac{2}{3}\lambda_{1}\lambda_{0}\wp_{2233}
- \frac{14}{3}\lambda_{1}\lambda_{0}\wp_{23}^2
- 4\lambda_{0}^2\wp_{23}\wp_{34}
+ 4\lambda_{2}^2\lambda_{0}\wp_{23} \\
&&\quad \textstyle + 4\lambda_{0}^2\wp_{24}\wp_{33}
+ 8\lambda_{0}\wp_{12}^2\wp_{44}
- \lambda_{2}\lambda_{1}^2\wp_{23}
+ 2\lambda_{3}\lambda_{1}^2\wp_{13}
- 2\lambda_{1}\wp_{11}\wp_{1244}  \\
&&\quad \textstyle - 16\lambda_{4}\lambda_{0}^2\wp_{23}
- 4\lambda_{0}\wp_{12}\wp_{1244}
+ 8\lambda_{3}\lambda_{0}\wp_{11}\wp_{33}
+ 4\lambda_{4}\lambda_{3}\lambda_{0}\wp_{12}\wp_{33}\\
&&\quad \textstyle + 8\lambda_{4}\lambda_{0}\wp_{12}\wp_{23}\wp_{33}
+ 2\lambda_{4}\lambda_{3}\lambda_{1}\wp_{11}\wp_{33}
+ 4\lambda_{4}\lambda_{1}\wp_{11}\wp_{23}\wp_{33}
+ \lambda_{1}^2\wp_{13}\wp_{23} \\
&&\quad \textstyle + \frac{16}{3}\lambda_{4}\lambda_{1}\lambda_{0}\wp_{13}
- 6\lambda_{3}\lambda_{2}\lambda_{0}\wp_{13}
+ \frac{16}{3}\lambda_{2}\lambda_{1}\lambda_{0}\wp_{34}
+ 4\lambda_{2}\lambda_{0}\wp_{12}\wp_{33} \\
&&\quad \textstyle - 4\lambda_{1}\lambda_{0}\wp_{14}\wp_{33}
- 5\lambda_{3}\lambda_{1}\lambda_{0}\wp_{23}
+ 4\lambda_{1}\lambda_{0}\wp_{13}\wp_{34}
+ \frac{2}{3}\lambda_{1}\lambda_{0}\wp_{22}\wp_{33} \\
&&\quad \textstyle + 4\lambda_{1}\wp_{11}\wp_{12}\wp_{44}
+ 16\lambda_{0}\wp_{12}\wp_{14}\wp_{24}
+ 4\lambda_{0}\wp_{11}\wp_{23}\wp_{33}
- \frac{2}{3}\lambda_{4}\lambda_{1}\wp_{11}\wp_{2333}   \\
&&\quad \textstyle + 8\lambda_{1}\wp_{11}\wp_{14}\wp_{24}
+ \frac{20}{3}\lambda_{4}\lambda_{2}\lambda_{1}\lambda_{0}
- 2\lambda_{1}^3\wp_{34}
+ 8\lambda_{0}^2\wp_{13}
- \frac{4}{3}\lambda_{4}\lambda_{0}\wp_{12}\wp_{2333} \\
&&\quad \textstyle - 18\lambda_{4}\lambda_{3}\lambda_{0}^2
- 2\lambda_{4}\lambda_{1}^3
- \frac{4}{3}\lambda_{1}^2\lambda_{0}
+ 8\lambda_{2}\lambda_{0}^2
\\
&&\bm{(-42)} \quad \hspace*{0.285in} \wp_{111}^2 = \textstyle 4\wp_{11}^3
+ 8\lambda_{0}^2\wp_{22}\wp_{33}
- 12\wp_{11}\lambda_{0}Q_{1244}
+ 8\lambda_{2}\lambda_{0}^2\wp_{34} \\
&&\quad \textstyle - 42\lambda_{3}\lambda_{0}^2\wp_{23}
- 4\lambda_{1}^2\wp_{11}\wp_{33}
- 4\lambda_{2}\lambda_{0}\wp_{13}^2
+ 2\lambda_{2}\lambda_{1}^2\wp_{13}
+ 8\lambda_{0}^2\wp_{13}\wp_{34} \\
&&\quad \textstyle - 4\lambda_{1}^2\lambda_{0}\wp_{34}
- 8\lambda_{0}^2\wp_{14}\wp_{33}
- 8\lambda_{2}^2\lambda_{0}\wp_{13}
+ 8\lambda_{4}\lambda_{0}^2\wp_{13}
+ \lambda_{1}^2\wp_{13}^2 - 8\lambda_{0}^2\wp_{23}^2 \\
&&\quad \textstyle + 16\lambda_{2}\lambda_{0}\wp_{11}\wp_{33}
+ 6\lambda_{3}\lambda_{1}\lambda_{0}\wp_{13}
- 4\lambda_{1}^3\wp_{23}
+ 16\lambda_{2}\lambda_{1}\lambda_{0}\wp_{23} + 2\lambda_{0}^2Q_{2233} \\
&&\quad \textstyle - 4\lambda_{1}\lambda_{0}\wp_{12}\wp_{33}
+ 4\lambda_{1}\lambda_{0}\wp_{13}\wp_{23}
+ 8\lambda_{4}\lambda_{0}\wp_{11}Q_{1444}
- 4\lambda_{4}\lambda_{1}^2\lambda_{0} \\
&&\quad \textstyle + 8\lambda_{4}\lambda_{2}\lambda_{0}^2
+ \lambda_{2}^2\lambda_{1}^2
- 4\lambda_{3}\lambda_{1}^3
- 4\lambda_{2}^3\lambda_{0}
- 27\lambda_{3}^2\lambda_{0}^2
+ 4\lambda_{1}\lambda_{0}^2
+ 18\lambda_{3}\lambda_{2}\lambda_{1}\lambda_{0}
\end{eqnarray*}
\end{theorem}
Some of these relations were presented in \cite{bego08} however the full set and the possible structure has not been determined until now.
\begin{proof}
Once again, these relations can be derived through a variety of methods as discussed in \cite{eeg10}.  They can again be found constructively using the $\sigma$-expansion, although in this case there is no simple linear algebra result to dictate that such relations exist, although such a result would surly follow from considering a basis for $\Gamma(6)$.  Instead we may simply search for them using arbitrary polynomials of the functions from the basis for $\Gamma(2)$.  The computations involved can be heavy and so simplifications are made by ensuring the polynomials are homogeneous in weight.  It is often possible to save computation by allocating the cubic terms to cancel higher order poles using Definition \ref{def:nip}.  Another method that minimises computation is to try and find the relations through manipulation of the bilinear relations discussed in the next section, (by differentiating them or multiplying them by a 3-index $\wp$-function).  However, it is not possible to find all the relations without using the $\sigma$-expansion.     \\
\end{proof}

In the hyperelliptic cases it is possible to write the set of quadratic relations concisely in the form a matrix equation, (see \cite{bel97} and \cite{MEeo10}).  Such equations can be seen to follow from the representation theoretic approach to Abelian functions developed by Athorne in \cite{CA2008}.  The corresponding results do not hold for the non-hyperelliptic cases but it is possible that similar results may be possible, perhaps using a family of matrix equations.  This is the subject of current research.

\subsection{Bilinear relations} \label{SEC_DE_bi}

The final set of differential equations that we consider here are a set bilinear in the 2 and 3-index $\wp$-functions.  Due to the parity properties of the $\wp$-functions we know that these \textit{bilinear relations} cannot contain any constant terms, or terms dependent only on the 2-index $\wp$-functions.  There is no analogue of these relations in the genus one case.  They have been considered in a variety of cases in \cite{bel97}, \cite{eemop07}, \cite{bego08}, \cite{MEe09}, \cite{MEhgt10} and \cite{MEeo10}.

The simplest way to construct these relations is through cross multiplication of the 4-index relations.  For example, substituting using equations (\ref{eq:35_p4444}) and (\ref{eq:35_p3444}) into
\[
\frac{\partial}{\partial u_3} \big( \wp_{4444} \big) - \frac{\partial}{\partial u_4} \big( \wp_{3444} \big) = 0
\]
gives the first relation  in Theorem \ref{thm:BL_35set} below.  This was the approach that was taken in the papers cited above.  However, the existence of $Q$-functions in the 4-index relations in the non-hyperelliptic cases means that more care has to be taken in the choice of cross products.  In higher genus trigonal cases, (or in the case where $n>3$) the inclusion of further Q-functions in the basis makes this method increasingly cumbersome.

An alternative method to systematically find bilinear relations has been developed making use of the $\mathcal{B}$-functions, one of the new generic classes of functions discussed in Section \ref{SEC_Bases}.  These functions are defined in equation (\ref{eq:BF}) as those linear combinations of $\{\wp_{ij}\wp_{lmn} \}$ which have poles of order at most three for a general curve.  We can search for bilinear relations as linear combinations of $\mathcal{B}$-functions along with the $\wp_{ijk}$-functions.  This approach requires more computation than taking cross products of 4-index relations.  However, as discussed in the proof below, it has the advantage of being systematic and allowing us to conclude that we have found a complete set of such relations.

\begin{theorem} \label{thm:BL_35set}
Every bilinear relation associated with the cyclic (3,5)-curve may be given as a linear combination of 65 relations starting with those below.  The full set is available online at \cite{DBAFweb}.
\begin{eqnarray*}
&&\bm{(-6)} \,\,\, \quad 0= \textstyle -\frac{1}{2}\wp_{2 4 4}
- \frac{1}{2}\wp_{3 3 3} - \wp_{3 4}\wp_{4 4 4}
+ \wp_{4 4}\wp_{3 4 4} \\
&&\bm{(-7)} \,\,\, \quad 0= \textstyle  2\wp_{2 3 4} - \lambda_{4}\wp_{3 4 4} - \wp_{3 3}\wp_{4 4 4}
- \wp_{3 4}\wp_{3 4 4} + 2\wp_{4 4}\wp_{3 3 4} \\
&&\bm{(-8)} \,\,\, \quad 0= \textstyle \frac{1}{3}\lambda_{4}\wp_{3 3 4} -\frac{2}{3}\wp_{2 3 3}
- \frac{2}{3}\wp_{3 3}\wp_{3 4 4} + \frac{1}{3}\wp_{3 4}\wp_{3 3 4}
+ \frac{1}{3}\wp_{4 4}\wp_{3 3 3}
\\ &&\quad \textstyle
- \wp_{2 4}\wp_{4 4 4} + \wp_{4 4}\wp_{2 4 4} \\
&&\bm{(-9)} \,\,\, \quad 0= \textstyle -2 \wp_{1 4 4} - \wp_{3 3}\wp_{3 3 4} + \wp_{3 4}\wp_{3 3 3}
+ \wp_{2 4}\wp_{3 4 4} - \wp_{3 4}\wp_{2 4 4} \\
&&\bm{(-9)} \,\,\, \quad 0= \textstyle 2\wp_{2 2 4}  -2\wp_{1 4 4} - \lambda_{4}\wp_{2 4 4}
- \wp_{2 3} \wp_{4 4 4} + \wp_{3 4}\wp_{2 4 4}
\\ &&\quad \textstyle
+ 2 \wp_{4 4}\wp_{2 3 4} - 2\wp_{2 4}\wp_{3 4 4} \\
&& \hspace*{0.80in} \vdots
\\
&&\bm{(-32)} \quad 0= \textstyle  \wp_{12}\wp_{111} - \wp_{11}\wp_{112}
+ \frac{1}{3}\lambda_{4}\lambda_{0}\big( 2\lambda_{4}\lambda_{0}\wp_{13}\wp_{334}
- 2\wp_{34}\wp_{133} \\
&&\quad \textstyle - 2\wp_{12}\wp_{444}
- 2\wp_{24}\wp_{144}
+ 4\wp_{44}\wp_{124}
\big) + \frac{1}{2}\lambda_{3}\lambda_{1}\big(
\wp_{13}\wp_{334}  - \wp_{34}\wp_{133}
\big) \\
&&\quad \textstyle + \lambda_{4}^2\lambda_{0}\big(
\frac{2}{3}\wp_{44}\wp_{224}
- \frac{1}{3}\wp_{22}\wp_{444}
- \frac{1}{3}\wp_{24}\wp_{244}
- \wp_{23}\wp_{334}
+ \wp_{34}\wp_{233}
\big) \\
&&\quad \textstyle + \frac{1}{9}\lambda_{4}^3\lambda_{0}\big(
2\lambda_{4}^3\lambda_{0}\wp_{34}\wp_{334}
+ 2\lambda_{4}^3\lambda_{0}\wp_{44}\wp_{333}
- 4\wp_{33}\wp_{344}
- \lambda_{4}^3\lambda_{0}\wp_{233}
\big) \\
&&\quad \textstyle + \frac{1}{3}\lambda_{0}\big(
4\wp_{14}\wp_{144}
- 4\wp_{44}\wp_{122}
- 4\wp_{24}\wp_{124}
+ 2\wp_{22}\wp_{144}
+ 5\wp_{23}\wp_{133} \\
&&\quad \textstyle + 5\wp_{12}\wp_{333}
- 8\wp_{44}\wp_{114}
- 10\wp_{33}\wp_{123}
+ 6\wp_{12}\wp_{244}
+ 4\wp_{11}\wp_{444}
\big) \\
&&\quad \textstyle + \lambda_{3}\lambda_{0}\big(
\wp_{22}\wp_{444}
+ \wp_{24}\wp_{244}
- 2\wp_{44}\wp_{224}
+ 3\wp_{23}\wp_{334}
- 3\wp_{34}\wp_{233}\\
&&\quad \textstyle + \frac{7}{3}\wp_{133}
\big)
+ \frac{1}{3}\lambda_{1}^2\big(
2\wp_{233}
- \wp_{34}\wp_{334}
+ 2\wp_{33}\wp_{344}
- \wp_{44}\wp_{333}
\big) \\
&&\quad \textstyle + \frac{1}{3}\lambda_{4}\lambda_{3}\lambda_{0}\big(
\wp_{233}
- 2\wp_{34}\wp_{334}
- 2\wp_{44}\wp_{333}
+ 4\wp_{33}\wp_{344}
\big) \\ &&\quad \textstyle  + \frac{1}{2}\lambda_{1}\big(
\wp_{12}\wp_{224} - \wp_{24}\wp_{122}
\big)
- \frac{4}{3}\lambda_{4}^2\lambda_{3}\lambda_{0}\wp_{334}
+ \frac{1}{6}\lambda_{4}\lambda_{1}^2\wp_{334} \\
&&\quad \textstyle + 2\lambda_{3}^2\lambda_{0}\wp_{334}
+ \frac{2}{9}\lambda_{4}^4\lambda_{0}\wp_{334}
+ \frac{25}{6}\lambda_{1}\lambda_{0}\wp_{334}
- 2\lambda_{2}\lambda_{0}\wp_{233}
- \frac{1}{2}\lambda_{2}\lambda_{1}\wp_{133}
\end{eqnarray*}
\end{theorem}
\noindent A number of these relations were originally derived in \cite{bego08} but this complete set was not available until now.
\begin{proof}
These relations were derived alongside an alternative basis for $\Gamma(3)$, constructed as follows:  We start by including the 3-index $\wp$-functions in the basis.  We then consider decreasing weights in turn and look to see how many $\mathcal{B}$-functions at that weight may be included in the basis.  Note that unlike the $\wp$-functions the order of the indices in the $\mathcal{B}$-function is relevant.  The definition of the $\mathcal{B}$-functions can simplify for specific indices and so we first test to see which $\mathcal{B}$-functions at a given weight are distinct.  We then test their linear dependence with themselves and with the $\wp_{ijk}$.  Note that when doing this we need not consider the entries in $\Gamma(2)$.  Although they are included in the basis they are all even functions and so cannot be linearly dependent on the $\mathcal{B}$-functions which are all odd.  We test the linear independence using the $\sigma$-expansion as discussed in the proof of Theorem \ref{thm:35_3pole}.  It is important to make use of weight simplifications when expanding products of series to reduce computation.

We find that the following 19 functions may be added to the basis, at the weights indicated.  Note that they play the role of the 19 $Q$-derivatives in the alternative basis in Theorem \ref{thm:35_3pole}.
\[
\begin{array}{lclcl}
\, \, \, \bm{(-8)}  \quad \mathcal{B}_{33344}      & \qquad & \bm{(-16)} \quad \mathcal{B}_{12334} \\
\, \, \, \bm{(-9)}  \quad \mathcal{B}_{23444}      & \qquad & \bm{(-17)} \quad \mathcal{B}_{12333}, \, \mathcal{B}_{12244}, \, \mathcal{B}_{11444} \\
\bm{(-11)} \quad \mathcal{B}_{23334}, \, \mathcal{B}_{22444}& \qquad & \bm{(-18)} \quad B_{11344} \\
\bm{(-12)} \quad \mathcal{B}_{13444}               & \qquad &  \bm{(-19)} \quad \mathcal{B}_{12233} \\
\bm{(-13)} \quad \mathcal{B}_{13344}               & \qquad &  \bm{(-20)} \quad \mathcal{B}_{12224},\,  \mathcal{B}_{11244} \\
\bm{(-14)} \quad \mathcal{B}_{22244}, \, \mathcal{B}_{13334}, \, \mathcal{B}_{12444} & \qquad & \bm{(-23)} \quad \mathcal{B}_{11144}  \\
\bm{(-15)} \quad \mathcal{B}_{12344} & &
\end{array}
\]
Now, any bilinear relation between the 2 and 3-index $\wp$-functions must have the order 4 and 5 poles in the $\wp_{ij}\wp_{klm}$-terms canceling.  These terms must be hence composed of $\mathcal{B}$-functions:  The only alternative would be if there were combinations with reduced orders of poles specific to the curve.  It is simple to check that no such combinations exist in the (3,5)-case as we can check using just the leading order terms of the $\sigma$-expansion, given in (\ref{eq:35SW}).  So all bilinear relations are given by linear combinations of $\mathcal{B}$-functions and $\wp_{ijk}$.  There will be one such relation for each of the distinct $\mathcal{B}$-functions not in the basis, and every other bilinear relation will be a combination of these.  The actual relations were derived as a by-product of the derivation of the alternative basis to $\Gamma(3)$. \\
\end{proof}

In general, this approach allows us to derive all independent bilinear relations and conclude we have the full set.  However, care must be taken in drawing the conclusion that all bilinear relations have pole cancelation in the form of $\mathcal{B}$-functions.  As discussed in the proof, we need to check that their are no cancelations specific to the curve.  In the (2,5)-case for example, there exists the function $\Delta$ in equation (\ref{eq:DF}) which has poles of order 3 in general but poles of order 2 on the (2,5)-curve.  Hence the derivatives of $\Delta$ are sums of $\wp_{ij}\wp_{klm}$ which are not $\mathcal{B}$-functions and yet have poles of order three.  This leads to \textit{extra} bilinear relations.  The existence of functions like $\Delta$ and the following extra bilinear relations appears to be a feature unique to the hyperelliptic cases.

We note that there are no bilinear relations at weights $-5$ or $-35$ in association with $\mathcal{B}_{44444}$ and $\mathcal{B}_{11111}$ as may be expected.  This is because the definition of the $\mathcal{B}$-functions collapses to zero in the case where all the indices are the same.  This also explains why there can be no bilinear relations in the elliptic (genus one) case.

\section{Addition Formulae} \label{SEC_AF}

Here we discuss the addition formulae satisfied by the Abelian functions and present a new formula associated with the cyclic trigonal curve of genus four.  We start by considering the formulae which generalise equation (\ref{eq:Intro_elliptic_add}) from the elliptic case.  Such a formula will express the ratio of functions
\[
\frac{\sigma(\bu + \bv)\sigma(\bu - \bv)}{\sigma(\bu)^2\sigma(\bv)^2}
\]
as a sum of products of pairs of functions from the basis for $\Gamma(2)$, one a function of $\bu$ and the other of $\bv$.  This must be so since the ratio above can be shown to be Abelian in both $\bu$ and $\bv$ and it has the correct pole structures, (see \cite{MEeo10} for details).  We can go further and use the parity property of $\sigma(\bu)$ to check that the function is either symmetric or anti-symmetric with respect to the change of variables $(\bu,\bv) \mapsto (\bv,\bu)$ when the $\sigma$-function is odd or even respectively.

The coefficients in the polynomial can be explicitly determined using the $\sigma$-expansion.  (See \cite{MEe09} for details of such calculations).  Such formulae have been found for a variety of cases in \cite{ba07}, \cite{bel97}, \cite{bego08}, \cite{MEe09}, \cite{MEhgt10}.  In particular, the formula for the cyclic (3,5)-curve was derived in Section 8 of \cite{bego08} and is given by
\[
\frac{\sigma(\bu + \bv)\sigma(\bu - \bv)}{\sigma(\bu)^2\sigma(\bv)^2} = f(\bu,\bv) + f(\bv,\bu)
\]
where
\begin{eqnarray*}
&&f(\bu,\bv) = \textstyle \wp_{44}(\bu)\wp_{11}(\bv) + \wp_{22}(\bu)\wp_{14}(\bv) - \wp_{12}(\bu)\wp_{24}(\bv) \\
&&\quad \textstyle - 2\wp_{14}(\bu)\wp_{14}(\bv)  - \frac{1}{2}Q_{1224}(\bu) + Q_{1244}(\bv)\wp_{34}(\bu)   \\
&&\quad \textstyle + \frac{1}{6}Q_{2233}(\bu)\wp_{33}(\bv)
- \frac{1}{3}Q_{2444}(\bv)\wp_{13}(\bu)
- \frac{1}{6}Q_{2333}(\bu)\wp_{23}(\bv) \\
&&\quad \textstyle + \frac{1}{3}Q_{2333}(\bv)\wp_{34}(\bu)\lambda_{4}
- \frac{2}{3}\wp_{13}(\bu)\wp_{33}(\bv)\lambda_{4}
- \frac{2}{3}\wp_{33}(\bu)\wp_{34}(\bv)\lambda_{2} \\
&&\quad \textstyle - \frac{1}{4}Q_{2333}(\bu)\lambda_{3}   - \wp_{33}(\bu)\wp_{34}(\bv)\lambda_{4}\lambda_{3}
+ \Big( \frac{5}{3}\lambda_{1} - \frac{1}{3}\lambda_{4}\lambda_{2} + \frac{3}{4}\lambda_{3}^2 \Big)\wp_{33}(\bu).
\end{eqnarray*}
Note that the (3,5) $\sigma$-function is even and hence the addition formula here is symmetric in $(\bu,\bv)$.  (The reason for the difference between this formula and the one is \cite{bego08} is that some formulae in \cite{bego08} were calculated using different $Q$-functions in the basis.  Our paper makes choices consistent with the basis derived in Section 5 of \cite{bego08}.)

\vspace{3mm}

In cyclic non-hyperelliptic cases there are more addition formulae associated with the functions, resulting from automorphisms of the curve equation.  Such addition formulae were the topic of \cite{emo11} which gave a thorough treatment of the genus one and two cases.  We will present a new genus four addition formulae associated with the cyclic (3,5)-curve.  It is related to the automorphism of the curve (\ref{eq:c35}) given by the operator
\[
[\zeta]: (x,y) \mapsto (x, \zeta y), \quad \mbox{where } \zeta=\exp\left( \frac{2\pi i}{3} \right).
\]
So $\zeta$ is a cube root of unity and $[\zeta]$ an operator which multiplies $y$ by the root leaving the curve unchanged.  We extend this notation to define the sequence of operators and automorphisms,
\[
[\zeta^j]: (x,y) \mapsto (x, \zeta^j y), \quad \mbox{for } j \in \mathbb{Z}.
\]
We can check using the basis of differentials (\ref{eq:du35}) that these operators act on the variables $\bu$ as follows.
\begin{equation} \label{eq:zetau}
[\zeta^j]\bu = ( \zeta^j u_1, \zeta^j u_2, \zeta^{2j} u_3, \zeta^j u_4 ).
\end{equation}
The action of such operators on the lattice $\Lambda$ is stable, (it moves the points around but does not change the overall lattice).  This can be checked by considering the effect on the individual elements of the period matrices, (see \cite{on98} for more details).  We can now derive the following result for the $\sigma$-function which follows Lemma 4.2.5 in \cite{on98}.

\begin{lemma} \label{lem:sigzeta}
The $\sigma$-function associated to the cyclic (3,5)-curve satisfies
\begin{equation} \label{sigma_zeta}
\sigma( [\zeta^j]\bu ) = \zeta^{2j}\sigma(\bu).
\end{equation}
\end{lemma}
\begin{proof}
Consider the quasi-periodicity of $\sigma( [\zeta^j]\bu )$.  If $\ell$ is a point on the lattice then
\[
\sigma( [\zeta^j](\bu+\ell) ) = \sigma( [\zeta^j]\bu + [\zeta^j]\ell ).
\]
Since the lattice is stable under the action we know that $[\zeta^j]\ell$ is also on the lattice.  Hence by (\ref{eq:HG_quas})
\[
\sigma( [\zeta^j](\bu+\ell) )
= \chi([\zeta^j]\ell) \sigma([\zeta^j]\bu) \exp \Big[ L \Big( [\zeta^j] \Big(\bu+\frac{\ell}{2}\Big), [\zeta]\ell \Big) \Big]
\]
In \cite{on98} the author shows that for an automorphism of a cyclic curve we have
\[
L([\zeta]\bu,\bv) = L(\bu,[\zeta^{-1}]\bv)
\]
and hence
\[
L([\zeta^j]\bu,[\zeta^j]\bv) = L(\bu,[\zeta^{-j}][\zeta^j]\bv) = L(\bu,\bv).
\]
Therefore, we have
\[
\sigma( [\zeta^j](\bu+\ell) )
= \chi([\zeta^j]\ell) \sigma([\zeta^j]\bu) \exp \Big[ L \Big( \bu+\frac{\ell}{2}, \ell \Big) \Big]
\]
We now consider the quotient
\[
\frac{\sigma( [\zeta^j](\bu+\ell) )}{\sigma( (\bu+\ell) )}
= \frac{\chi([\zeta^j]\ell)}{\chi(\ell)} \frac{\sigma([\zeta^j]\bu)}{\sigma(\bu)}
= \pm \frac{\sigma([\zeta^j]\bu)}{\sigma(\bu)}
\]
since $\chi(\ell) = \pm 1$. So we see that the function
\[
\frac{\sigma([\zeta^j]\bu)}{\sigma(\bu)}
\]
is bounded and entire (since the zero sets coincide).  Hence, by Liouville's theorem, the function is a constant.  Using the Schur-Weierstrass polynomial (\ref{eq:35SW}) we see that this constant is $\zeta^{2j}$.  \\
\end{proof}

We can now derive the addition formula associated with these automorphisms.

\begin{theorem} \label{thm:35_3t2v}
The functions associated to the cyclic (3,5)-curve satisfy
\[
\frac{\sigma(\bu+\bv)\sigma(\bu+[\zeta]\bv)\sigma(\bu+[\zeta^2]\bv) }
{ \sigma(\bu)^3 \sigma(\bv)^3 } = f(\bu,\bv) + f(\bv,\bu)
\]
where
\[
f(\bu,\bv) = \Big[ P_{24} + P_{21} + P_{18} + P_{15} + P_{12} + P_{9} + P_{6} + P_{3} + P_{0} \Big](\bu,\bv)
\]
and the polynomials $P_k(\bu,\bv)$ are as presented in Appendix \ref{APP_35add}.
\end{theorem}
\begin{proof}
Denote the left hand side of the formula by LHS$(\bu,\bv)$.  Using Lemma \ref{lem:sigzeta} and the parity property of the $\sigma$-function we first check that LHS$(\bu,\bv)$ is symmetric under $(\bu,\bv) \mapsto (\bv,\bu)$.  Next consider the affect of $\bu \mapsto \bu + \ell$.
\begin{eqnarray*}
&&\mbox{LHS}(\bu+\ell,\bv)
= \frac{\sigma(\bu+\ell+\bv)\sigma(\bu+\ell+[\zeta]\bv)\sigma(\bu+\ell+[\zeta^2]\bv)}
{\sigma(\bu+\ell)^3\sigma(\bv)^3} \\
&&\quad = \mbox{LHS}(\bu,\bv)\frac{\chi(\ell)e^{\Psi[\bu+\bv+\ell/2]}
\chi(\ell)e^{\Psi[\bu+[\zeta]\bv+\ell/2]}
\chi(\ell)e^{\Psi[\bu+[\zeta^2]\bv+\ell/2]}}
{\chi(\ell)^3e^{3\Psi[\bu+\ell/2]}} \\
&&\quad = \mbox{LHS}(\bu,\bv) e^{\Psi\bv[1 + [\zeta] + [\zeta^2]]}.
\end{eqnarray*}
However, from (\ref{eq:zetau}) we see that
\[
\bv(1 + [\zeta] + [\zeta^2]) =
\left( \begin{array}{l}
v_1 (1 + \zeta + \zeta^2 ) \\
v_2 (1 + \zeta + \zeta^2 ) \\
v_3 (1 + \zeta^2 + \zeta )
\end{array} \right) =
\left( \begin{array}{l}
0 \\ 0 \\ 0
\end{array} \right),
\]
and so LHS$(\bu,\bv)$ is Abelian with respect to $\bu$.  Further, since it is symmetric in $(\bu,\bv)$ we can conclude that it is Abelian in $\bv$ also.  It is now clear that LHS$(\bu,\bv)$ may be expressed as
\[
\mbox{LHS}(\bu,\bv) = \sum_i c_i A_i(\bu)B_i(\bv)
\]
where the $c_i$ are constants and the functions $A_i(\bu),B_i(\bv)$ belong to the basis for
$\Gamma(3)$, presented earlier in Theorem \ref{thm:35_3pole}.

To determine the constants $c_i$ we use the $\sigma$-expansion.  The computations involved can be heavy and so it is essential that we take into account all the available simplifications.  We have already noted that $\mbox{LHS}(\bu,\bv)$ is symmetric in $(\bu,\bv)$, but we can also check that it is even with respect to $(\bu,\bv) \mapsto (-\bu,-\bv)$.  Hence we need only consider combinations of functions that preserve these two properties.  Further, we can check that it has total weight $-24$ and so only consider those combinations with this weight (including when combined with monomials in the curve constants).  Together these simplifications drastically reduces the number of possible terms.  To further ease the time and memory constraints we implement code in Maple to efficiently expand the products of series so that only the relevant terms are considered.

We find that $\mbox{LHS}(\bu,\bv)$ is given as stated.  For simplicity we group together the terms with common weight ratios into the polynomials $P_k(\bu,\bv)$ which contain the terms with weight $-k$ in the Abelian functions and weight $-(24+k)$ in the curve parameters.  The polynomials are presented in Appendix \ref{APP_35add}. \\
\end{proof}

\noindent This formula is actually a reduction of a more general formula in three sets of variables expressing
\[
\frac{\sigma(\bu+\bv+\bw)\sigma(\bu+[\zeta]\bv+[\zeta^2]\bw)\sigma(\bu+[\zeta^2]\bv+[\zeta]\bw) }
{ \sigma(\bu)^3 \sigma(\bv)^3\sigma(\bw)^3 }
\]
as a sum of products of triplets of functions from $\Gamma(3)$, one each in $\bu,\bv$ and $\bw$.  The existence of these formulae were predicted in \cite{bego08} but we can only now derive them since they require the basis for $\Gamma(3)$, calculated in Section \ref{SEC_Bases} using the new classes of Abelian functions that are the subject of this paper.

\vspace{3mm}

We finish by noting that there are further addition formula associated with further reductions of the (3,5)-curve.  The curve $y^3 = x^5 + \lambda_0$ will have a family of automorphisms given by
\[
[i^j]: (x,y) \mapsto (\eta^jx,y)
\]
where $\eta$ is a fifth root of unity.  We can check that the ratio
\[
\frac{\sigma(\bu+\bv)\sigma(\bu+[\eta]\bv)\sigma(\bu+[\eta^2]\bv)\sigma(\bu+[\eta^3]\bv)\sigma(\bu+[\eta^4]\bv) }{ \sigma(\bu)^5 \sigma(\bv)^5 }
\]
is Abelian in both $\bu$ and $\bv$ and so a similar approach could be employed to derive the formula, following the derivation of a basis for $\Gamma(5)$.  However both these computations would be computationally intensive.

\vspace{3mm}

\appendix

\section{Generic Abelian functions with poles of order four} \label{APP_Pole4}

In this Appendix we repeat the procedure described in Section \ref{ss_classes} to calculate Abelian function which have poles of order four for any $(n,s)$-curve.

The first non trivial case occurs when we allow five indices.  We may construct polynomials from 5-index $\wp$-functions and the products of a 2 and 3-index $\wp$-function.  There is only one independent non-zero combination which has poles of order four, given by
\[
\wp_{ijklm} + \wp_{ij}\wp_{klm} - 13\wp_{lm}\wp_{ijk}.
\]
(We note that the coefficient 13 may seem unexpected but it can be easily verified that it is necessary to balance the poles of order 5 in the other two terms.)

With six indices we can build polynomials from the set of functions
\[
\{ \wp_{ijklmn}, \wp_{ijk}\wp_{lmn}, \wp_{ij}\wp_{klmn}, \wp_{ij}\wp_{kl}\wp_{mn} \}.
\]
There are three independent non-zero combinations which have poles of order 4.  These are given by
\begin{eqnarray*}
&&\wp_{ijklmn} \textstyle + \wp_{ij}\wp_{klmn} - \frac{1}{2}\big(
3\big(\wp_{jk}\wp_{ilmn} + \wp_{jl}\wp_{ikmn} + \wp_{jm}\wp_{ikln} \\
&& \quad \textstyle + \wp_{jn}\wp_{iklm} \big) + 5\big(\wp_{kl}\wp_{ijmn} + \wp_{km}\wp_{ijln} + \wp_{kn}\wp_{ijlm} + \wp_{lm}\wp_{ijkn} \\
&& \quad \textstyle + \wp_{ln}\wp_{ijkm} + \wp_{mn}\wp_{ijkl} \big)
\big),
\\
&&\wp_{ij}\wp_{klmn} \textstyle - \frac{1}{4}\big(
3\wp_{ijk}\wp_{lmn} + 3\wp_{ijl}\wp_{kmn}  + 3\wp_{ijm}\wp_{kln} + 3\wp_{ijn}\wp_{klm}
\\ && \quad \textstyle - \wp_{ikl}\wp_{jmn} - \wp_{ikm}\wp_{jln} - \wp_{ikn}\wp_{jlm}
- \wp_{ilm}\wp_{jkn}  - \wp_{iln}\wp_{jkm}  - \wp_{imn}\wp_{jkl}
\big),
\\
&&\wp_{ij}\wp_{kl}\wp_{mn} \textstyle - \frac{1}{8} \big(
\wp_{ijk}\wp_{lmn} + \wp_{ijl}\wp_{kmn} + \wp_{ijm}\wp_{kln} + \wp_{ijn}\wp_{klm}+ \wp_{ikl}\wp_{jmn}
\\ && \quad \textstyle - \wp_{ikm}\wp_{jln} - \wp_{ikn}\wp_{jlm}
- \wp_{ilm}\wp_{jkn} - \wp_{iln}\wp_{jkm} + \wp_{imn}\wp_{jkl}
\big).
\end{eqnarray*}
Note that we take linear combinations of these to define the class,
\begin{eqnarray*}
\mathcal{F}_{ijklmn} &=&  \wp_{ij}\wp_{kl}\wp_{mn} - \wp_{ij}\wp_{kn}\wp_{lm} - \wp_{il}\wp_{jk}\wp_{mn} + \wp_{il}\wp_{jn}\wp_{km} \\
&& \qquad + \wp_{im}\wp_{jl}\wp_{kn}
- \wp_{im}\wp_{jn}\wp_{kl} + \wp_{in}\wp_{jk}\wp_{lm} - \wp_{in}\wp_{jl}\wp_{km},
\end{eqnarray*}
which were used to complete the basis for $\Gamma(4)$ in the (3,4)-case as detailed in \cite{MEeo10}.
This procedure may be continued indefinitely.  We have compiled functions with up to 8 indices and these can be found online at \cite{DBAFweb}.  One class of particular note is
\begin{eqnarray*}
&&\mathcal{G}_{ijklmnop} = \wp_{ijklmnop} - 4\big( \wp_{ijno}\wp_{klmp} + \wp_{ijnp}\wp_{klmo} + \wp_{ijlp}\wp_{kmno} \\
&&\quad + \wp_{ijmn}\wp_{klop} + \wp_{ijmo}\wp_{klnp} + \wp_{ijmp}\wp_{klno} + \wp_{ijkm}\wp_{lnop} + \wp_{ijkn}\wp_{lmop} \\
&&\quad + \wp_{ijko}\wp_{lmnp} + \wp_{ijkp}\wp_{lmno} + \wp_{ijlm}\wp_{knop} + \wp_{ijln}\wp_{kmop} + \wp_{ijlo}\wp_{kmnp} \\
&&\quad + \wp_{ijkl}\wp_{mnop} + \wp_{ijop}\wp_{klmn} + \wp_{iklm}\wp_{jnop} + \wp_{ikln}\wp_{jmop} + \wp_{iklo}\wp_{jmnp} \\
&&\quad + \wp_{iklp}\wp_{jmno} + \wp_{ikmn}\wp_{jlop} + \wp_{ikmo}\wp_{jlnp} + \wp_{ikmp}\wp_{jlno} + \wp_{ikno}\wp_{jlmp} \\
&&\quad + \wp_{iknp}\wp_{jlmo} + \wp_{ikop}\wp_{jlmn} + \wp_{ilmn}\wp_{jkop} + \wp_{ilmo}\wp_{jknp} + \wp_{ilmp}\wp_{jkno} \\
&&\quad + \wp_{ilno}\wp_{jkmp} + \wp_{ilnp}\wp_{jkmo} + \wp_{ilop}\wp_{jkmn} + \wp_{imno}\wp_{jklp} + \wp_{imnp}\wp_{jklo} \\
&&\quad + \wp_{imop}\wp_{jkln} + \wp_{inop}\wp_{jklm} \big).
\end{eqnarray*}
which was used to complete the basis for $\Gamma(4)$ in the (2,7)-case as detailed in \cite{MEeo10}.

\section{The 4-index relations} \label{APP_4index}

The full set of 4-index relations for the functions associated with the cyclic (3,5)-curve are given below.
\begin{eqnarray*}
&&\bm{(-4)} \quad \wp_{4444} = \textstyle  6\wp_{44}^{2} - 3\wp_{33}
\\
&&\bm{(-5)} \quad \wp_{3444} = \textstyle  6\wp_{34}\wp_{44} + 3\wp_{24}
\\
&&\bm{(-6)} \quad \wp_{3344} = \textstyle  4 \wp_{34}^{2} - \wp_{23} + 2\wp_{34}\lambda_{4} + 2\wp_{33}\wp_{44}
\\
&&\bm{(-7)} \quad \wp_{3334} = \textstyle  6\wp_{33}\wp_{34} - Q_{{2444}}
\\
&&\bm{(-7)} \quad \wp_{2444} = \textstyle  6\wp_{24}\wp_{44} + Q_{{2444}}
\\
&&\bm{(-8)} \quad \wp_{2344} = \textstyle  4\wp_{14} - \wp_{22} + 2\wp_{24}\lambda_{4}+ 2\wp_{23}\wp_{44} + 4\wp_{24}\wp_{34}
\\
&&\bm{(-8)} \quad \wp_{3333} = \textstyle  12\wp_{14} - 3\wp_{22} + 6 \wp_{33}^{2}
\\
&&\bm{(-9)} \quad \wp_{2334} = \textstyle  2\wp_{13} + 3\wp_{34}\lambda_{3}+ \lambda_{2} + 4\wp_{23}\wp_{34} + 2\wp_{24}\wp_{33}
\\
&&\bm{(-10)} \quad \wp_{2244} = \textstyle  \frac{2}{3}Q_{{2444}}\lambda_{4}+ \wp_{33}\lambda_{3}+ 2\wp_{22}\wp_{44} + 4\wp_{24}^{2}+ \frac{2}{3}Q_{{1444}}
\\
&&\bm{(-10)} \quad \wp_{1444} = \textstyle  Q_{{1444}}+ 6\wp_{14}\wp_{44}
\\
&&\bm{(-10)} \quad \wp_{2333} = \textstyle  6\wp_{23}\wp_{33} - 2Q_{{1444}}+ 3\wp_{33}\lambda_{3}
\\
&&\bm{(-11)} \quad \wp_{2234} = \textstyle  4\wp_{23}\wp_{24} - 2\wp_{12} + 4\wp_{14}\lambda_{4}+ 3\wp_{24}\lambda_{3}- 2\wp_{44}\lambda_{2}+ 2\wp_{22}\wp_{34}
\\
&&\bm{(-11)} \quad \wp_{1344} = \textstyle  4\wp_{14}\wp_{34} - \wp_{12} + 2\wp_{14}\lambda_{4}+ 2\wp_{13}\wp_{44}
\\
&&\bm{(-12)} \quad \wp_{1334} = \textstyle  2\wp_{13}\lambda_{4} + \frac{3}{2}\wp_{23}\lambda_{3} + 2\wp_{34}\lambda_{2}
+ \lambda_{4}\lambda_{2} + 4\wp_{13}\wp_{34} \\
&&\qquad \textstyle + 2\wp_{14}\wp_{33} - \frac{1}{2}Q_{{2233}}
\\
&&\bm{(-12)} \quad \wp_{2233} = \textstyle  Q_{{2233}}+ 2\wp_{22}\wp_{33} + 4 \wp_{23}^{2}
\\
&&\bm{(-13)} \quad \wp_{1333} = \textstyle  6\wp_{13}\wp_{33} - 2\lambda_{4}Q_{{1444}}+ 3Q_{{1244}}
\\
&&\bm{(-13)} \quad \wp_{1244} = \textstyle  Q_{{1244}} + 2\wp_{12}\wp_{44} + 4\wp_{14}\wp_{24}
\\
&&\bm{(-13)} \quad \wp_{2224} = \textstyle  4\lambda_{4}Q_{{1444}}+ Q_{{2444}}\lambda_{3}+ 6\wp_{33}\lambda_{2}- 6Q_{{1244}} + 6\wp_{22}\wp_{24}
\\
&&\bm{(-14)} \quad \wp_{2223} = \textstyle  6\wp_{22}\wp_{23} - 6\wp_{11} + 6\wp_{14}\lambda_{3}+ 3\wp_{22}\lambda_{3}- 6\wp_{44}\lambda_{1}
\\
&&\bm{(-14)} \quad \wp_{1234} = \textstyle  2\wp_{14}\wp_{23} - \wp_{11} + 3\wp_{14}\lambda_{3}- \wp_{44}\lambda_{1}+ 2\wp_{12}\wp_{34} + 2\wp_{13}\wp_{24}
\\
&&\bm{(-15)} \quad \wp_{1233} = \textstyle  3\wp_{13}\lambda_{3}+ 2\wp_{34}\lambda_{1}+ \lambda_{4}\lambda_{1}- 3\lambda_{0}+ 2\wp_{12}\wp_{33} + 4\wp_{13}\wp_{23}
\\
&&\bm{(-16)} \quad \wp_{2222} = \textstyle 4Q_{{2444}}\lambda_{2} - 12Q_{{1244}}\lambda_{4}
+ 8{\lambda_{4}}^{2}Q_{{1444}} + 2\lambda_{3}Q_{{1444}} \\
&&\qquad \textstyle - 12\wp_{33}{\lambda_{3}}^{2} + 12\wp_{33}\lambda_{4}\lambda_{2}+ 6\wp_{33}\lambda_{1}
+ 6Q_{{1144}}+ 3Q_{{2333}}\lambda_{3}+ 6 \wp_{22}^{2}
\\
&&\bm{(-16)} \quad \wp_{1144} = \textstyle  Q_{{1144}}+ 2\wp_{11}\wp_{44} + 4 \wp_{14}^{2}
\\
&&\bm{(-16)} \quad \wp_{1224} = \textstyle  - Q_{{1144}}- \frac{1}{2}Q_{{2333}}\lambda_{3} + \frac{3}{2}\wp_{33} {\lambda_{3}}^{2}+ 3\wp_{33}\lambda_{1}\\
&&\qquad \textstyle + 4\wp_{12}\wp_{24} + 2\wp_{14}\wp_{22}
\\
&&\bm{(-17)} \quad \wp_{1223} = \textstyle  - 2\wp_{11}\lambda_{4}+ 3\wp_{12}\lambda_{3} + 4\wp_{14}\lambda_{2}- 2\wp_{24}\lambda_{1}
- 6\wp_{44}\lambda_{0} \\
&&\qquad \textstyle + 4\wp_{12}\wp_{23} + 2\wp_{13}\wp_{22}
\\
&&\bm{(-17)} \quad \wp_{1134} = \textstyle  2\wp_{14}\lambda_{2}- \wp_{24}\lambda_{1} + 2\wp_{11}\wp_{34} + 4\wp_{13}\wp_{14}
\\
\end{eqnarray*}
\begin{eqnarray*}
&&\bm{(-18)} \quad \wp_{1133} = \textstyle  2\wp_{11}\wp_{33} + 4 \wp_{13}^{2}+ 2\wp_{13}\lambda_{2} - \wp_{23}\lambda_{1}+ 6\wp_{34}\lambda_{0}
\\
&&\bm{(-19)} \quad \wp_{1124} = \textstyle  \frac{2}{3}\lambda_{2}Q_{{1444}}- \frac{1}{3}Q_{{2444}}\lambda_{1} + 4\wp_{33}\lambda_{0}+ 2\wp_{11}\wp_{24} + 4\wp_{12}\wp_{14}
\\
&&\bm{(-19)} \quad \wp_{1222} = \textstyle  - 3Q_{{1244}}\lambda_{3}+ 2\lambda_{4}\lambda_{3}Q_{{1444}}+ 2Q_{{2444}}\lambda_{1}
+ 6\wp_{33}\lambda_{4}\lambda_{1} \\
&&\qquad \textstyle + 6\wp_{33}\lambda_{0}+ 6\wp_{12}\wp_{22}
\\
&&\bm{(-20)} \quad \wp_{1123} = \textstyle  2\wp_{12}\lambda_{2}- 6\wp_{24}\lambda_{0}- \wp_{22}\lambda_{1}+ 4\wp_{14}\lambda_{1}+ 2\wp_{11}\wp_{23} + 4\wp_{12}\wp_{13}
\\
&&\bm{(-22)} \quad \wp_{1122} = \textstyle  \wp_{33}\lambda_{3}\lambda_{1}+ 8\wp_{33}\lambda_{4}\lambda_{0}+ \frac{4}{3}\lambda_{4}\lambda_{2}Q_{{1444}}
- 2Q_{{1244}}\lambda_{2} \\
&&\qquad \textstyle + 2Q_{{2444}}\lambda_{0} + 2\wp_{11}\wp_{22} + 4 \wp_{12}^{2}+ \frac{2}{3}\lambda_{1}Q_{{1444}}
\\
&&\bm{(-22)} \quad \wp_{1114} = \textstyle  - 2Q_{{2444}}\lambda_{0}+ 6\wp_{11}\wp_{14} + \lambda_{1}Q_{{1444}}
\\
&&\bm{(-23)} \quad \wp_{1113} = \textstyle  6\wp_{11}\wp_{13}- 6\wp_{22}\lambda_{0}+ 6\wp_{14}\lambda_{0}+ 3\wp_{12}\lambda_{1}
\\
&&\bm{(-25)} \quad \wp_{1112} = \textstyle  2\lambda_{0}Q_{{1444}}+ 6\wp_{33}\lambda_{3}\lambda_{0}
+ 2\lambda_{4}\lambda_{1}Q_{{1444}}+ 6\wp_{11}\wp_{12}- 3Q_{{1244}}\lambda_{1}
\\
&&\bm{(-28)} \quad \wp_{1111} = \textstyle 6 \wp_{11}^{2} - 3\wp_{33} {\lambda_{1}}^{2} + 8\lambda_{4}\lambda_{0}Q_{{1444}}
+ 12\wp_{33}\lambda_{2}\lambda_{0}- 12Q_{{1244}}\lambda_{0}
\end{eqnarray*}
These relations are also available online at \cite{DBAFweb}.

\section{The addition formula} \label{APP_35add}

The addition formula presented in Theorem \ref{thm:35_3t2v} is constructed using the following polynomials.
\begin{eqnarray*}
&&P_{24}(\bu,\bv) = \textstyle
- \frac{1}{8}\mathcal{T}_{222222}(\bv)
+ \frac{3}{4}\mathcal{T}_{114444}(\bv)\wp_{23}(\bu)
+ \frac{3}{4}\mathcal{T}_{133344}(\bv)\wp_{13}(\bu) \\
&&\quad \textstyle + \frac{3}{8}\mathcal{T}_{222224}(\bu)\wp_{34}(\bv)
+ \frac{1}{8}\wp_{122}(\bv)\wp_{224}(\bu)
+ \frac{1}{4}\wp_{114}(\bu)\wp_{224}(\bv) \\
&&\quad \textstyle - \frac{1}{8}\wp_{124}(\bu)\wp_{124}(\bv)
- \frac{1}{4}\wp_{144}(\bu)\wp_{114}(\bv)
- \frac{1}{8}\wp_{333}(\bu)\wp_{112}(\bv) \\
&&\quad \textstyle + \frac{1}{4}\wp_{122}(\bu)\wp_{144}(\bv)
- \frac{1}{8}\wp_{222}(\bu)\wp_{124}(\bv)
- \frac{3}{8}\wp_{112}(\bu)\wp_{244}(\bv) \\
&&\quad \textstyle + \frac{1}{4}\wp_{444}(\bu)\wp_{111}(\bv)
+ \frac{1}{16}Q_{2233}(\bu)Q_{2233}(\bv)
- \frac{1}{24}\partial_{3}Q_{1244}(\bu)\partial_{3}Q_{2444}(\bv) \\
&&\quad \textstyle - \wp^{[1313]}(\bv)\wp_{23}(\bu)
+ \frac{1}{2}\wp^{[1323]}(\bv)\wp^{[2334]}(\bu)
+ \frac{1}{4}\wp^{[1323]}(\bv)\wp_{13}(\bu)\\
&&\quad \textstyle + \frac{1}{2}\wp^{[3434]}(\bv)\wp^{[1313]}(\bu)
- \frac{1}{2}\wp^{[2323]}(\bu)\wp^{[1334]}(\bv)
- \frac{1}{2}\wp^{[1334]}(\bu)\wp^{[1334]}(\bv) \\
&&\quad \textstyle - \frac{1}{8}\partial_{3}Q_{1144}(\bv)\wp_{333}(\bu)
- \frac{1}{24}\wp_{122}(\bu)\partial_{3}Q_{2444}(\bv)
+ \frac{1}{24}\wp_{222}(\bu)\partial_{3}Q_{1444}(\bv) \\
&&\quad \textstyle + \frac{1}{8}\partial_{3}Q_{1244}(\bv)\wp_{224}(\bu)
- \frac{1}{2}\partial_{3}Q_{1244}(\bu)\wp_{144}(\bv)
+ \frac{1}{6}\partial_{3}Q_{2444}(\bu)\wp_{114}(\bv) \\
&&\quad \textstyle - \frac{1}{72}\partial_{3}Q_{1444}(\bu)\partial_{3}Q_{1444}(\bv)
- \frac{1}{6}\wp_{124}(\bv)\partial_{3}Q_{1444}(\bu)
\\ \\
&&P_{21}(\bu,\bv) = \textstyle \big( \,
\frac{1}{2}\wp_{13}(\bu)\wp^{[1334]}(\bv)
- \frac{1}{4}\wp_{13}(\bv)Q_{2233}(\bu)
- \wp_{34}(\bu)\wp^{[1313]}(\bv) \\
&&\quad \textstyle - \frac{1}{12}\wp_{224}(\bu)\partial_{3}Q_{1444}(\bv)
+ \frac{1}{3}\wp_{144}(\bv)\partial_{3}Q_{1444}(\bu)
+ \frac{1}{4}\wp_{114}(\bu)\wp_{333}(\bv) \\
&&\quad \textstyle + \frac{1}{36}\partial_{3}Q_{2444}(\bv)\partial_{3}Q_{1444}(\bu)
+ \frac{3}{8}\mathcal{T}_{222224}(\bu)
\, \big) \lambda_{4}
\\ \\
\end{eqnarray*}
\begin{eqnarray*}
&&P_{18}(\bu,\bv) = \textstyle \big( \,
\frac{3}{4}\mathcal{T}_{114444}(\bv)
- \frac{3}{8}\mathcal{T}_{133344}(\bv)\wp_{34}(\bu)
- \frac{1}{8}\wp_{333}(\bu)\wp_{124}(\bv) \\
&&\quad \textstyle - \frac{3}{8}\wp_{23}(\bu)Q_{2233}(\bv)
- \wp^{[1313]}(\bv)
+ \frac{3}{8}\wp^{[1323]}(\bu)\wp_{34}(\bv) \\
&&\quad \textstyle + \frac{3}{4}\wp_{13}(\bu)\wp_{13}(\bv)
+ \frac{1}{24}\wp_{333}(\bv)\partial_{3}Q_{1444}(\bu)
\, \big) \lambda_{3}
- \frac{1}{2}\wp^{[1313]}(\bu)\lambda_{4}^2
\\ \\
&&P_{15}(\bu,\bv) = \textstyle \big( \,
\frac{1}{4}\wp^{[1323]}(\bv)
+ \frac{1}{4}\wp_{333}(\bv)\wp_{144}(\bu)
+ \frac{3}{8}Q_{2233}(\bv)\wp_{34}(\bu) \\
&&\quad \textstyle - \wp_{34}(\bu)\wp^{[1334]}(\bv)
- \frac{3}{4}\wp_{34}(\bu)\wp^{[2323]}(\bv)
+ \frac{1}{2}\wp_{23}(\bu)\wp_{13}(\bv) \\
&&\quad \textstyle + \frac{1}{2}\wp^{[3434]}(\bu)\wp_{13}(\bv)
+ \frac{3}{4}\mathcal{T}_{133344}(\bv)
\, \big)\lambda_{2} \\
&&\quad \textstyle + \big( \,
\frac{3}{4}\wp_{23}(\bu)\wp_{13}(\bv)
- \frac{3}{8}\mathcal{T}_{133344}(\bv)
- \frac{1}{8}\wp^{[1323]}(\bu) \\
&&\quad \textstyle - \frac{3}{4}\wp_{34}(\bu)\wp^{[1334]}(\bv)
- \frac{1}{4}\wp_{34}(\bu)Q_{2233}(\bv)
\, \big)\lambda_{4}\lambda_{3}
\\ \\
&&P_{12}(\bu,\bv) = \textstyle
\frac{1}{4}Q_{2233}(\bv)\lambda_{3}^2
+ \frac{1}{2}\wp^{[2323]}(\bv)\lambda_{1}
- \frac{1}{2}Q_{2233}(\bu)\lambda_{1} \\
&&\quad \textstyle + \frac{1}{2}\wp_{34}(\bu)\wp^{[2334]}(\bv)\lambda_{1}
- \frac{1}{8}\wp_{244}(\bu)\wp_{333}(\bv)\lambda_{1}
+ \frac{5}{2}\wp_{34}(\bv)\wp_{13}(\bu)\lambda_{1} \\
&&\quad \textstyle + \frac{9}{16}\wp_{23}(\bu)\wp_{23}(\bv)\lambda_{3}^2
- \frac{1}{4}\wp_{23}(\bu)\wp^{[3434]}(\bv)\lambda_{1}
+ \frac{1}{8}Q_{2233}(\bv)\lambda_{4}\lambda_{2}\\
&&\quad \textstyle - \frac{3}{4}\wp^{[1334]}(\bu)\lambda_{4}^2\lambda_{3}
- \frac{1}{4}\wp^{[2323]}(\bu)\lambda_{4}\lambda_{2}
- \frac{3}{8}\wp^{[2334]}(\bv)\wp_{34}(\bu)\lambda_{3}^2  \\
&&\quad \textstyle - \frac{1}{4}Q_{2233}(\bv)\lambda_{4}^2\lambda_{3}
+ \frac{1}{2}\wp^{[1334]}(\bu)\lambda_{4}\lambda_{2}
- \frac{3}{2}\wp_{13}(\bu)\wp_{34}(\bv)\lambda_{3}^2\\
&&\quad \textstyle - \frac{1}{2}\wp^{[1334]}(\bu)\lambda_{3}^2
- \frac{1}{8}\wp^{[2323]}(\bu)\lambda_{3}^2
- \frac{1}{2}\wp_{34}(\bv)\wp_{13}(\bu)\lambda_{4}\lambda_{2}\\
&&\quad \textstyle - \frac{1}{4}\wp_{23}(\bu)\wp_{23}(\bv)\lambda_{1}
+ \wp_{34}(\bu)\wp_{13}(\bv)\lambda_{4}^2\lambda_{3}
\\ \\
&&P_{9}(\bu,\bv) = \textstyle \big( \,
\frac{3}{4}\wp_{34}(\bv)\wp_{23}(\bu)
+ \frac{1}{4}\wp_{333}(\bv)\wp_{444}(\bu)
- \frac{9}{2}\wp_{13}(\bv)
\, \big)\lambda_{0}  \\
&&\quad \textstyle \big( \,
\frac{1}{2}\wp^{[2334]}(\bu)
+ 3\wp_{13}(\bv)
- \frac{3}{4}\wp^{[3434]}(\bv)\wp_{34}(\bu)
- \frac{3}{4}\wp_{23}(\bv)\wp_{34}(\bu)
\, \big)\lambda_{3}\lambda_{2} \\
&&\quad \textstyle \big( \,
2\wp_{13}(\bu)
+ \frac{1}{2}\wp^{[2334]}(\bu)
+ \frac{5}{4}\wp_{34}(\bv)\wp_{23}(\bu)
\, \big)\lambda_{4}\lambda_{1}
- \frac{1}{2}\wp_{13}(\bu)\lambda_{4}^2\lambda_{2} \\
&&\quad \textstyle \big( \,
\frac{3}{4}\wp_{23}(\bv)\wp_{34}(\bu)
- 2\wp_{13}(\bu)
- \frac{3}{8}\wp^{[2334]}(\bu)
\, \big) \lambda_{4}\lambda_{3}^2
+ \wp_{13}(\bu)\lambda_{4}^3\lambda_{3}
\\ \\
&&P_{6}(\bu,\bv) = \textstyle \big( \,
\lambda_{3}\lambda_{1}
- \frac{1}{2}\lambda_{2}^2
+ \frac{3}{4}\lambda_{4}\lambda_{0}
\, \big)\wp_{23}(\bv) + \big( \,
\frac{3}{4}\lambda_{4}^2\lambda_{3}^2
- \frac{3}{4}\lambda_{3}^3
+ \lambda_{4}^2\lambda_{1}
\, \big)\wp_{23}(\bu) \\
&&\quad \textstyle
+ \wp^{[3434]}(\bu)\lambda_{2}^2
- \frac{1}{4}\wp^{[3434]}(\bu)\lambda_{3}\lambda_{1}
- \frac{3}{4}\wp^{[3434]}(\bv)\lambda_{4}\lambda_{3}\lambda_{2} \\
&&\quad \textstyle + \big( \,
\frac{1}{2}\lambda_{2}^2
- \frac{3}{2}\lambda_{4}^2\lambda_{1}
- 3\lambda_{4}\lambda_{0}
+ \frac{1}{4}\lambda_{4}\lambda_{3}\lambda_{2}
\, \big)\wp_{34}(\bu)\wp_{34}(\bv)
\\ \\
&&P_{3}(\bu,\bv) = \textstyle \big( \,
\frac{9}{4}\lambda_{3}\lambda_{0}
- \frac{3}{2}\lambda_{4}^2\lambda_{0}
- \frac{1}{2}\lambda_{3}^2\lambda_{2}
- \frac{3}{2}\lambda_{4}^3\lambda_{1}
- \frac{1}{4}\lambda_{4}\lambda_{2}^2
\, \big) \wp_{34}(\bv) \\
&&\quad \textstyle + \big( \,
\frac{3}{4}\lambda_{4}^2\lambda_{3}\lambda_{2}
+ \frac{5}{4}\lambda_{4}\lambda_{3}\lambda_{1}
\, \big) \wp_{34}(\bu)
\\ \\
&&P_{0}(\bu,\bv) = \textstyle
- \frac{1}{2}\lambda_{4}^2\lambda_{2}^2
+ \frac{1}{2}\lambda_{3}^2\lambda_{1}
+ \frac{1}{2}\lambda_{4}^3\lambda_{0}
- \frac{9}{4}\lambda_{2}\lambda_{0}
- \frac{1}{2}\lambda_{4}\lambda_{3}^2\lambda_{2}
+ \frac{1}{2}\lambda_{4}^3\lambda_{3}\lambda_{2} \\
&&\quad \textstyle - \frac{1}{4}\lambda_{4}^2\lambda_{3}\lambda_{1}
+ \frac{9}{4}\lambda_{4}\lambda_{2}\lambda_{1}
+ \frac{3}{2}\lambda_{4}\lambda_{3}\lambda_{0} - 2\lambda_{1}^2
\end{eqnarray*}
These polynomials are also available online at \cite{DBAFweb}.

\bibliography{DBAF}{}
\bibliographystyle{plain}

\end{document}